\documentclass[journal]{IEEEtran}
\usepackage[cmex10]{amsmath}
\usepackage{amsmath}
\usepackage{multicol}
\usepackage{amsfonts}
\usepackage{amssymb}
\usepackage{graphicx}
\usepackage{cite}
\usepackage{euscript}
\newtheorem{example}{Example}[section]
\newtheorem{definition}{Definition}
\newtheorem{corollary}{Corollary}
\newtheorem{fact}{Fact}
\newtheorem{lem}{Lemma}
\newtheorem{theorem}{Theorem}

\begin{document}

\title{Constrained Multi-user Multi-server\\ Max-Min Fair Queuing$^*$\thanks{* Work in progress, version 1.}}
\author{Jalal~Khamse-Ashari, Ioannis~Lambadaris, Yiqiang Zhao\thanks{Khamse-Ashari and Lambadaris are with the SCE department, Zhao is with the School of Mathematics and Statistics, Carleton University, Ottawa, Canada.}}
\maketitle

\begin{abstract}
In this paper, a multi-user multi-server queuing system is studied in which each user is constrained to get service from a subset of servers. In the studied system, rate allocation in the sense of max-min fairness results in multi-level fair rates. To achieve such fair rates, we propose $CM^4FQ$ algorithm. In this algorithm users are chosen for service on a packet by packet basis. The priority of each user $i$ to be chosen at time $t$ is determined based on a parameter known as service tag (representing the amount of work counted for user $i$ till time $t$). Hence, a free server will choose to serve an eligible user with the minimum service tag. Based on such simple selection criterion, $CM^4FQ$ aims at guaranteed fair throughput for each demanding user without explicit knowledge of each server service rate. We argue that $CM^4FQ$ can be applied in a variety of practical queuing systems specially in mobile cloud computing architecture.
\end{abstract}

\begin{IEEEkeywords}
Fair Queuing, multi-level Max-Min Fairness, Packet scheduling, Mobile Cloud Computing
\end{IEEEkeywords}

\section{Introduction}
\PARstart{M}{ax-min} Fairness has been known as the most prevalent notion of fairness in the concept of resource-allocation. In allocating one type of resource to some demanding users, an allocation is said to be max-min fair, if it is feasible and an increase in the amount of allocated resource to any arbitrary user, results in decreasing the amount of allocated resource to some other user(s) with smaller or equal allocation \cite{DN}.

If the resource can be divided into any arbitrary parts and each user is eligible to use any part of the resource, max-min fair allocation results in equal shares for different users. When the resource comprises of only one indivisible part (one server), time sharing or ``scheduling" is the natural way for multiplexing. Considering a set of greedy users, max-min fairness could be defined in this case based on the users' long-term/average usage of the resource \cite{DN}.

When the resource is arbitrarily divisible, max-min fair allocation could be defined in terms of instantaneous usage. Specifically, consider resources such as bandwidth or CPU. When such resources are arbitrarily divisible, max-min fair allocation could be defined in terms of \emph{instantaneous rates}. This approach has been known as Weighted Fair Sharing (WFS). By definition, in weighted fair sharing equal weighted rates are attributed to demanding users at any time $t$ \cite{FS}.

The most important property of WFS is being memoryless, i.e., each user's current rate of service is independent of its rate of service in the past. This is specially important when users have intermittent demands. According to memoryless property, if user $i$ has not any demand of the resource during interval $(t_0,t_1)$, its share of service is not reserved \cite{FS}, \cite{GPS}. Instead, it could be divided among other demanding users. The memoryless property is central in guaranteeing that users will not be starved in a work-conserving system.

When the resource comprises of only one indivisible part (i.e., one server), WFS could be realized if users' demands are infinitesimally divisible, i.e. represented by \emph{fluid flow} \cite{FS, GPS}. When users' demand or ``input traffic" is packetized, WFS could be achieved only approximately through packet by packet scheduling schemes\cite{SCFQ}. Many Fair Queuing (FQ) algorithms have been proposed to approximate WFS on a packet by packet scheduling basis. Self-Clocked Fair Queuing (SCFQ) is the first practically implementable algorithm achieving a good approximation of WFS \cite{SCFQ}.


In this paper, a system model is considered in which one type of resource, like CPU or bandwidth, is available on multiple servers. In addition, it is assumed that each user is eligible to get service only from a subset of servers. Considering
such a \emph{constrained} multi-user multi-server queuing system, max-min fair rate allocation results in \emph{multi-level} fair rates.

As an example, consider the queuing system shown in Fig.\ref{fig2}. Assume traffic stream of users to be fluid flow. According to max-min fairness definition (mentioned in the beginning of the introduction), the max-min fair rate of user $c$ equals to $r_c(t)=1$ Mbps, while the max-min fair rate of users $a$ and $b$ equal to $r_a(t)=r_b(t)=1.25$ Mbps. In such allocation, the service rate of server $1$ is equally divided between users $a$ and $b$, while the total service rate of server $2$ is allocated to user $c$.

\begin{figure}
\centering
\includegraphics[width=2.5in]{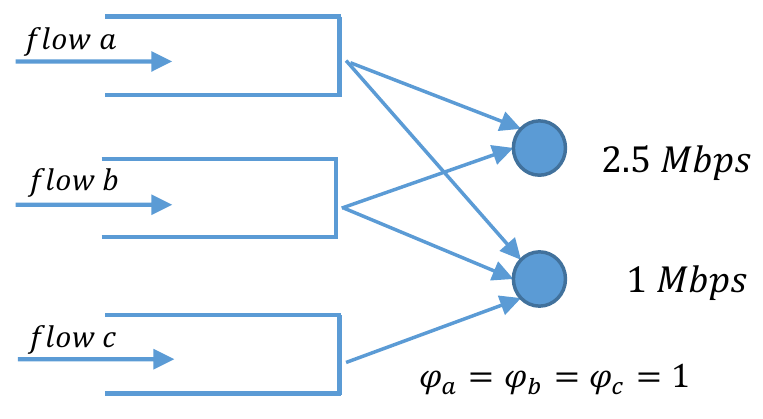}
\caption{A sample queuing system; Users $a$ and $b$ are eligible to get service from both servers, while user $c$ is only eligible to get service from server 2.}
\label{fig2}
\end{figure}

While there exists a rich literature on FQ, e.g., \cite{GPS}, \cite{SCFQ}, \cite{SFQ}, \cite{msFQ},\cite{Chandra00,Chandra01}, none of algorithms developed in the context of FQ is applicable for the case of constrained multi-user multi-server queuing system. Indeed, all existing FQ algorithms are designed to achieve equal fair rates among demanding users. For instance, while \cite{msFQ} has considered a multi-server queuing system, it has considered the case that any user is eligible to get service from \emph{any} server. As a consequent, in such fully connected system all users attain the same fair service rate. Hence, the proposed FQ algorithm in \cite{msFQ} is essentially the same as a single server FQ algorithms.

Regarding multi-level max-min FQ, the literature is very sparse. To the best of our knowledge, there is not any prior work which addresses multi-level max-min FQ problem in general. We have seen prior work in \cite{midrr} which has studied the same constrained multi-user multi-server queuing system. After presenting our proposed FQ algorithm and showing its generality, we will compare it with the results of \cite{midrr} in Section VII.

Considering the constrained multi-user multi-server queuing model, in this paper we propose a FQ algorithm which we call $CM^4FQ$\footnote{$CM^4FQ$ stands for Constrained Multi-user Multi-server Max-Min Fair Queuing.}. This algorithm is designed to achieve multi-level max-min fair rates. Like in single server FQ algorithms, in $CM^4FQ$ packets scheduling decisions are made simply based on the users' service tag. Service tag of each user $i$ represents the amount of work counted for user $i$ till time $t$. Given that server $k$ becomes free at time $t$, it will choose to serve the user with the minimum service tag among users eligible to get service on this server. Based on such simple selection criterion, $CM^4FQ$ aims at guaranteed fair throughput for each demanding user without explicit knowledge of each server service rate.

While $CM^4FQ$ addresses a general system model, it is very simple. Hence, it can be applicable in different scenarios. For example, in mobile cloud computing networks \cite{MCC}, multiple servers/machines are available for giving service to different applications/queues. Each server possibly has different amounts of resources, including CPU, RAM, etc. In correspondence with our queuing model, CPU could be considered as the scarce resource which is to be scheduled among users. On the other hand, each application possibly has different requirements of RAM, storage, etc., which may be available only on a subset of servers. These requirements could be reflected as certain constraints for each application to get service from different servers.

Another application conforming to the described system model, is related to smart mobile phones which have several interfaces for data access \cite{midrr}. Specifically, mobile phones may have data access via WiFi, 3G or 4G at the same time. On the other hand, there are certain restrictions for mobile applications to use different access techniques. For example, it may be preferred to stream video only over WiFi because it is cheaper, but to use the more expensive available data rate of 3G for VoIP. Reflecting such preferences as constraints on choosing different interfaces by each application, our proposed algorithm could be used for fair packet scheduling.

The rest of this paper is organized as follows. Basic definitions and system model are stated in section II. Our proposed FQ algorithm is described in section III. Section IV and V are devoted to characterizing properties and performance of the proposed FQ algorithm. Our analytical results are evaluated through numerical studies in section VI. After a brief review of related works in section VII, our conclusions are drawn in section VIII.

\section{Basic Definitions and System Model}
\subsection{System Model}
Consider a multi-user multi-server queuing system as shown in Fig. \ref{fig1} in which each user is constrained to get service from a subset of servers. Assume that $\cal{N}$$=\{1,2,...,N\}$ is the set of users and $\cal{K}$$=\{1,2,...,K\}$ is the set of servers. The service rate function of each server $k$ is $\rho^k(t)$ bits/sec. Let a binary matrix denoted by $\Pi=[\pi_{i,k}]$, determines the constraints of users to get service from different servers. Therefor, user $i$ is $\emph{eligible}$ to use server $k$ when $\pi_{i,k}=1$, and is not eligible otherwise. There exists at least a single $``1"$ element in each row and each column of the matrix $\Pi$, i.e., any user (server) is connected to at least one server (user).

To present the notion of max-min fairness in this scenario, first consider users' input traffic streams as fluid flow. Assuming an idealized fluid flow system, some basic definitions are presented in the next subsection.

\begin{figure}
\centering
\includegraphics[width=2in]{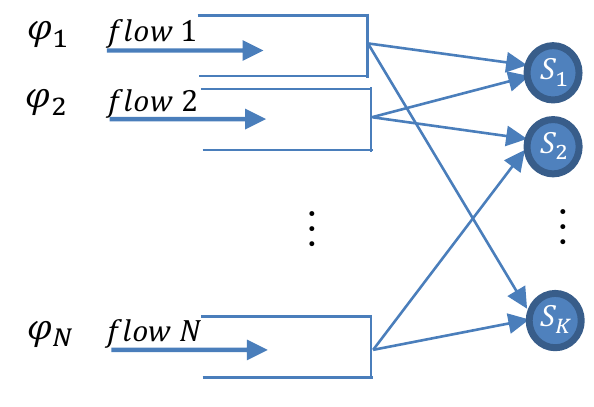}
\caption{System model (any user eligible to get service from a server is connected to that by an arrow)}
\label{fig1}
\end{figure}

\subsection{Basic Definitions}

\begin{definition}\label{def_1}
Any arbitrary user $i$ is said to be \emph{backlogged} at time $t$, if its corresponding queue is not empty at time $t^-$. The set of backlogged users at time $t$ is denoted by $B(t)$.
\end{definition}
\begin{definition}
In an idealized fluid flow system conforming to the system model of Fig. \ref{fig1}, let $r_{i,k}(t)$ be the service rate allocated to user $i$ by server $k$. The matrix $R(t)=[r_{i,k}(t)]_{N\times K}$ satisfying the following conditions is defined as a \emph{rate allocation matrix}.
\begin{itemize}
  \item[a)] If user $i$ is not backlogged at time $t$, then $r_{i,k}(t)=0$, $\forall k\in\cal{K}$.
  \item[b)] For any user $i\in \cal{N}$ and server $k\in\cal{K}$, if $\pi_{i,k}=0$ then $r_{i,k}(t)=0$.
  \item[c)] If at least one of users eligible to get service from server $k$ is backlogged at time $t$, then $\sum_{i\in\cal{N}}r_{i,k}(t)=\rho^k(t)$.
\end{itemize}
\end{definition}

The above conditions guarantee the queuing system to be work conserving. According to the rate allocation matrix $R(t)=[r_{i,k}(t)]$, the total service rate allocated to an arbitrary user $i$ at time $t$ is equal to $r_i(t)=\sum_{k\in\cal{K}}r_{i,k}(t)$.

Suppose that a positive wight $\varphi_i>1$ is associated to each user $i$. Using the weight associated to user $i$, the normalized service rate of user $i$ is defined as $r_i(t)/\varphi_i$.

\begin{definition}
A rate allocation mechanism is said to be \emph{memoryless} if the service rate of each user at time $t$ is independent of its service rate for all times prior to $t$, \cite{GPS}.
\end{definition}

\begin{definition}\emph{($CM^4$ fairness)}
A rate allocation matrix $R(t)$ is said to satisfy \emph{Constrained Multi-user Multi-server Max-Min ($CM^4$) fairness}, when increasing the normalized service rate of any user results in decreasing the normalized service rate of some other user(s) with smaller or equal normalized service rate.
\end{definition}

It should be mentioned that $CM^4$ fairness is just the general definition of max-min fairness which has been applied to the rate allocation in the constrained multi-user multi-server queuing system \cite{DN, midrr}. Since $CM^4$ fair rate allocation is defined in terms of instantaneous rates, it is obviously memoryless.

\begin{theorem}\label{th1}
Consider a fluid flow system conforming to the system model of Fig. \ref{fig1}. A rate allocation matrix $R(t)$ satisfies $CM^4$ fairness, if and only if statement (\ref{eq1}) holds for any backlogged user $j$:
\begin{equation}\label{eq1}
  \text{if }r_{i,k}(t)>0\text{ and }\pi_{j,k}=1 \Rightarrow \frac{r_j(t)}{\phi_j}\geq \frac{r_i(t)}{\phi_i}
\end{equation}
\end{theorem}

Theorem \ref{th1}, that is a concise statement of Theorem 2 in \cite{midrr}, offers a necessary and sufficient condition on $R(t)$ to be fair in conjunction with $CM^4$ fairness definition. Since the proof follows the same line of arguments as what provided for Theorem 2 in \cite{midrr}, it is omitted here. To get a better intuition on the concept of $CM^4$ fairness, we consider an example.

\begin{example}
Consider the queuing system shown in Fig. \ref{fig}. Assume traffic stream of users to be fluid flow. According to $CM^4$ fairness definition, the fair service rate of user $d$ equals to 1 Mbps, while the fair service rate of users $a$, $b$ and $c$ equals to 1.2 Mbps. To achieve such fair service rates, the total service rate of server 3 should be allocated to user $d$. However, there exist several ways for allocating the service rate of server 1 and server 2 to users $a$, $b$ and $c$,  while each of them gets a service rate of 1.2 Mbps.

The reader can verify that in this example there are more than one rate allocation matrices satisfying $CM^4$ fairness.
However, all of such matrices share some common properties. First, in all cases users $a$, $b$ and $c$ get the same service rate. Second, in all of them users with higher service rate are not given service on the server allocated to the user with lower service rate. Based on these observations, the set of users and servers could be classified into two clusters (Fig. \ref{fig}), satisfying the following conditions. I) Users in each cluster have the same fair service rates. II) All users in a cluster are only given service by servers in that cluster.
In the following, the concept of clustering is defined in general for any arbitrary configuration of users and servers.
\begin{figure}
\centering
\includegraphics[width=3.2in]{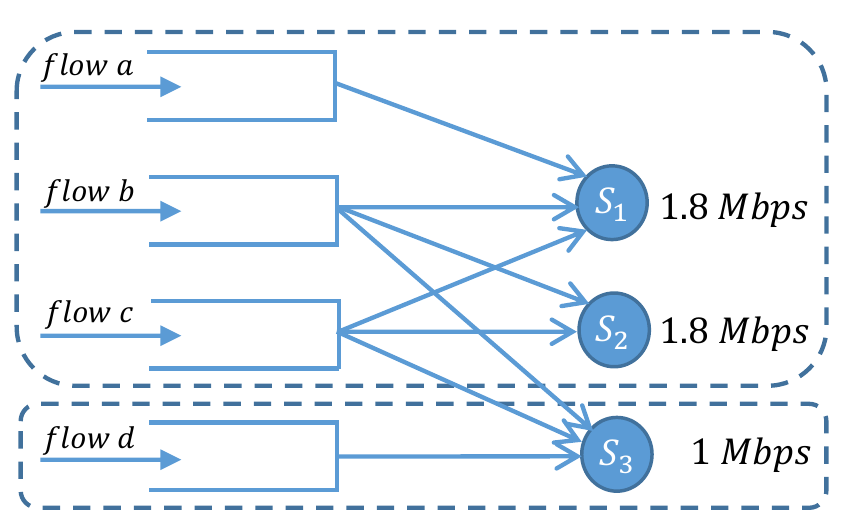}
\caption{A sample queuing system, the weight of all users are assumed to be equal. Here there are two clusters, $C_1=(\{a,b,c\},\{s_1,s_2\})$, $C_2=(\{d\},\{s_3\})$. The normalized service rate of users in cluster 1 is $r^{C_1}=1.2$ Mbps, and the normalized service rate of user $d$ in cluster 2 is $r^{C_2}=1$ Mbps.}
\label{fig}
\end{figure}
\end{example}

A cluster $C(t)$ is defined as an ordered pair $C(t)=(I(t),S(t))$, where $I(t)$ and $S(t)$ are non-empty subset of users and non-empty subset of servers, respectively. Suppose that all users are partitioned into $M$ disjoint subsets $I_m(t), 1\le m\le M$ and all servers are partitioned into the same number $M$ disjoint subsets $S_m(t), 1\le m\le M$. A clustering consisting of $M$ clusters is defined as $\{C_m(t)\}_{m=1}^M$, where $C_m(t)=(I_m(t),S_m(t))$.

\begin{definition}\label{def_FOC}
Suppose the rate allocation matrix $R(t)$ satisfies $CM^4$ fairness when $B(t)=B$ is the set of backlogged users at time $t$. Clustering $\{C_m(t)\}_{m=1}^M$ is said to be a \emph{Fairness Oriented Clustering} (FOC) in conjunction with $B$ if:

\begin{itemize}
  \item[a)] For any cluster $C_m(t)$, any user $i\in I_m(t)$ has the same normalized service rate, $r_i(t)/\varphi_i=r^{C_m}(t)$. We define $r^{C_m}(t)$ as the cluster service rate.
  \item[b)] For any two distinct clusters $C_m(t)$ and $C_l(t)$, $m\ne l$, $r^{C_m}(t)\neq r^{C_l}(t)$.
  \item[c)] Server $k$ belongs to cluster $C_m(t)$ with $r^{C_m}(t)>0$, if and only if $r_{i,k}(t)>0$ for some user $i\in I_m(t)$.
\end{itemize}
\end{definition}

\begin{lem}\label{lem1}
For given set of backlogged users and service rates of servers, there exists a unique FOC.
\end{lem}

For a given set of backlogged users at time $t$, $B(t)=B$, it may be possible that more than one rate allocation matrices satisfying $CM^4$ fairness exist. However, Lemma \ref{lem1} implies that all of them result in the same unique FOC. The proof of Lemma \ref{lem1} can be found in the appendix. A similar clustering method referred to as \emph{rate clustering} is presented in \cite{midrr}. Indeed, the FOC satisfies the \emph{rate clustering property}. However, the rate clustering does not result in a unique clustering for a given set of backlogged users.

\begin{fact}\label{cor1}
For given service rate of servers, FOC at time $t_0$ could be viewed as a function of the set of backlogged users $B_0=B(t_0)$, i.e., $\{C_m(B_0)\}_{m=1}^M$.
\end{fact}

It is straightforward to show that for any two clusters $m,l$ such that $r^{C_m}>r^{C_l}$, any user $i$ in $C_l$ can not be eligible for getting service from servers $k$ in cluster $C_m$ (i.e., $\pi_{i,k}=0$). This is formally stated in the following corollary.

\begin{corollary}\label{cor1.1}
Let $\{C_m(t)\}_{m=1}^M$ be the FOC which corresponds to the set of backlogged users $B=B(t)$. If some user $j\in I_m(t)$ with positive service rate is eligible to get service from server $k$ in another cluster, $\pi_{j,k}=1$ and $k\in S_l(t)$, then it follows that $r^{C_m}(t)>r^{C_l}(t)$.
\end{corollary}


\section{FQ Algorithm}
An FQ algorithm that achieves $CM^4$ fairness is proposed in this section. Although in the introductory material presented thus far we mainly focused on (idealized) fluid flow system, the algorithm that we will develop operates on packet by packet basis. In particular, when a server becomes free, the algorithm decides on which packet to be served by that server, while respecting users service constraints (determined by matrix $\Pi$). Any packet chosen to get service from a server, gets all of its service from that server without any interruption.

Since the algorithm works on a packet-by-packet basis, each user could achieve its fair service rate in an ``approximate" fashion. More specifically, consider an interval $(t_1,t_2)$ during which the set of backlogged users does not change. According to $CM^4$ fairness definition, for the same set of backlogged users in an idealized fluid flow system, each backlogged user has a specific fair service rate, and therefor each backlogged user has a specific fair service (work) share for that interval. In Section V, it is shown that the difference between the amount of work offered in $[t_1,t_2)$ to any backlogged user by $CM^4FQ$ algorithm in the actual system vs its fair share within $[t_1,t_2)$ in an idealized fluid flow system with the same set of backlogged users remains bounded.

Before presenting the algorithm, we introduce some notation along with assumptions.

\subsection{Notation and Assumptions}
\begin{itemize}
  \item It is assumed that each server can hold only one packet under service. When user $i$ is chosen for service by server $k$, its head of queue packet is moved to the server. The server keeps the packet till its service is completed.
  \item It is assumed that service rate of each server $k$ is time invariant\footnote{Since the proposed algorithm works without explicit knowledge of each server service rate, it could be used even when servers have time varying service rates. However, for the sake of simplifying the presentation we make such assumption.}, i.e., $\rho^k(t)=\rho^k$ for all $t$.
  \item According to Definition \ref{def_1}, a user is said to be backlogged at time $t$, if its corresponding queue is not empty at time $t^-$. Hence, a user with empty queue and some packets under service is \emph{not} considered to be backlogged.
  \item The set of backlogged users at time $t$ which are eligible to get service from server $k$ is denoted by $B^k(t)$:
  \begin{equation}\label{eq2}
    B^k(t)\triangleq\{i\mid i\in B(t)\text{ and }\pi_{i,k}=1\}.
  \end{equation}
  Server $k$ is said to be backlogged at time $t$, if $B^k(t)\ne\emptyset$.
  \item All functions of time in this paper are assumed to be left continuous, i.e.,: $f(t^-)=f(t)$.
  \item For any scalar function $f(t)$, $f(t_1,t_2)$ is defined as $f(t_1,t_2)=f(t_2)-f(t_1)$.
  \item It is assumed that packets of each user have a maximum length of $L^{max}$ bits. The length of each packet in bits is considered as the required work of that packet.
\end{itemize}

\subsection{The $CM^4FQ$ Algorithm}\label{alg}
$CM^4FQ$ algorithm belongs to the class of FQ algorithms. In such algorithms, the priority of each user $i$ to be chosen at time $t$ is determined based on a parameter, $F_i(t)$, which is computed through the algorithm and will be referred as \emph{service tag}. $CM^4FQ$ assigns an individual service tag to each user\footnote{In correspondence with FQ algorithms like SFQ or SCFQ \cite{SFQ,SCFQ}, $F_i(t)$ represents ``Finish Service Tag" of the last packet of user $i$ which has been chosen to get service before time $t$ From another perspective, for user $i$ being backlogged at time $t$, $F_i(t)$ represents ``Start Service Tag" of its head of queue packet.}; $F_i(t)$ represents the amount of work counted for user $i$ through the execution of the algorithm till time $t$. $F_i(0)$ is initially set as zero for all users. At time $t$ a free server will choose to serve an eligible user with the minimum service tag.

For any server $k$ a function $V^k(t)$ is defined to represent the \emph{work level} of server $k$ at time $t$. For server $k$ with nonempty set of backlogged users ($B^k(t)\ne\emptyset$), $V^k(t)$ will be set to the minimum service tag of users $i\in B^k(t)$. For the case where the set $B^k(t)$ is empty, $V^k(t)$ will be set to $\infty$. Therefor, $V^k(t)$ is defined by:

\begin{equation}\label{eq_V}
 V^k(t)\triangleq\left \{ \begin{array}{l}
\min\{F_j(t)\mid j\in B^k(t)\}\quad \text{if }B^k(t)\ne\emptyset\\
\infty \qquad\qquad\qquad\qquad\qquad\text{otherwise}
 \end{array}\right.
\end{equation}

Intuitively, $V^k(t)$ gives a measure of work progress in the system till time $t$ from the perspective of server $k$. Work level of servers at time $t$, will be used in updating the service tag of an idle user which becomes backlogged at time $t$.

We will now present the algorithm which is summarized in table \ref{table1}. In order to get intuition and understand its de-queuing mechanism, we assume that server $k$ becomes free at time $t$. Among backlogged users eligible to get service from this server, the head of queue packet of user $i^*$ with the minimum of $F_i(t)$ is chosen for service (as stated in \eqref{eq5}). In case of a tie, the user with the minimum index is chosen. Given that packet $p$ of user $i^*$ is chosen for service, $F_{i^*}(t^+)$ is incremented by $L_p/\varphi_{i^*}$ (as stated in \eqref{eq6}), where $L_p$ is the length of packet $p$ in bits, and $\varphi_{i^*}$ as stated previously represents the weight for user $i^*$ \footnote{For the sake of simplicity, the time index for all parameters will be omitted in the pseudo-code representation of the algorithm.}.

\begin{table}
\footnotesize
\caption{The pseudo code of $CM^4FQ$ Algorithm}
\label{table1}
\begin{tabular}{p{8.48cm}}
\hline\noalign{\smallskip}
\bfseries   $CM^4FQ$ Algorithm\\
\hline\noalign{\smallskip}
\begin{itemize}
  \item \textbf{Initialization}
  \begin{eqnarray}
    \nonumber
    &&\text{For }(i=1;\text{ }i\le N;\text{ }i++)\qquad\qquad\qquad\qquad\qquad\qquad\\ \nonumber
    &&\qquad F_i=D_i=0;\\ \nonumber
    &&\text{For }(k=1;\text{ }k\le K;\text{ }k++)\\ \nonumber
    &&\qquad V^k=D^k=0;
  \end{eqnarray}
  \item \textbf{En-queuing Module}: On arrival of packet $p$ of user $i$
  \begin{eqnarray}
      \nonumber
      &&\text{If }(Empty(Queue_i))\\\nonumber
      &&\qquad {Activate\text{-}Servers}(i)\\
      &&\qquad F_i=\max\{F_i,\max\{V^k\mid\pi_{i,k}=1\}\};\label{eq4}\\\nonumber
      &&Enqueue(i,p);\\\nonumber
      &&\text{//Enqueue packet $p$ to queue of user $i$ and update $B$}
  \end{eqnarray}
  \item \textbf{De-queuing Module(server k)}
  \begin{eqnarray}
  \nonumber &&\text{While}(B^k\neq\emptyset)\\
    &&\qquad i^*\leftarrow\text{argmin}_i\{F_i\mid i\in B^k\};\label{eq5}\\\nonumber
    &&\qquad p=Dequeue(Head(Queue_{i^*}));\\\nonumber
    &&\qquad\text{//Dequeue head of queue of user $i$ and update $B$}\\\nonumber
    &&\qquad L_p=Size(p);\\
    &&\qquad F_{i^*}=F_{i^*}+L_p/\varphi_{i^*};\label{eq6}\\\nonumber
    &&\qquad {UpdateV}(i^*,k);\\\nonumber
    &&\qquad Service(k,p); \text{  //Give service to packet $p$ on server $k$}
  \end{eqnarray}
\end{itemize}\\
\noalign{\smallskip}\hline
\end{tabular}
\end{table}

\begin{table}
\footnotesize
\caption{The subroutines of $CM^4FQ$ Algorithm}
\label{table2}
\begin{tabular}{p{8.48cm}}
    \hline\noalign{\smallskip}
    \begin{itemize}
\item \textbf{ActivateServers$(i)$}
    \begin{eqnarray}
        \nonumber
        &&\qquad V^0=0;\\\nonumber
        &&\qquad\text{If }(\{k \mid V^k<\infty\}\ne\emptyset)\\
        &&\qquad\qquad V^0=\max\{V^k\mid V^k<\infty\}+\delta;\\\label{eq3_0}\nonumber
        &&\qquad\text{For }(l=1;\text{ }l\le K\text{ \& } \pi_{i,l}=1\text{ \& }V^l=\infty;\text{ }l++) \\
        &&\qquad\qquad V^l=V^0\label{eq3_1};
    \end{eqnarray}
\item \textbf{UpdateV$(i,k)$}
    \begin{eqnarray}
  \nonumber
  && \text{//updating work level of servers}\\\nonumber
  && \hat{V}^{k}=V^{k};\\ \nonumber
  && \text{For }(l=1;\text{ }l\le K\text{ \& }\pi_{i,l}=1;\text{ }l++) \\
  && \qquad V^l=\infty;\label{eq8_1}\\\nonumber
  && \qquad \text{If }(B^l\ne\emptyset)\\
  && \qquad\qquad V^l=\min\{F_j\mid j\in B^l\}\label{eq8_2}\\\nonumber
  &&\text{//regulating the gap among work level of servers}\\\nonumber
  && d^k=0;\\\nonumber
  &&\text{If }(\{l\mid V^l<\hat{V}^k\}\ne\emptyset)\\\nonumber
  &&\qquad d^k=\min\{V^l\mid \hat{V}^k\leq V^l\}\\
  &&\qquad\qquad\!   -\max\{V^l\mid V^l<\hat{V}^k\}-\delta;\,\label{eq9_0}\\\nonumber
  && \text{If }(d^k>0)\\ \nonumber
  && \qquad\text{For }(j=1;\text{ }j\le N\text{ \& }F_j\geq\hat{V}^k;\text{ }j++)\\
  && \qquad\qquad F_j = F_j-d^k \label{eq9_1};\\
  && \qquad\qquad D_j = D_j+d^k \label{eq9_2};\\ \nonumber
  && \qquad\text{For }(l=1;\text{ }l\le K\text{ \& }V^l\geq\hat{V}^k;\text{ }l++)\\
  && \qquad\qquad V^l = V^l-d^k \label{eq9_3};\\
  && \qquad\qquad D^l = D^l+d^k\label{eq9_4};
\end{eqnarray}
\end{itemize}\\
\noalign{\smallskip}\hline
\end{tabular}
\end{table}

If user $i^*$ is chosen for service by server $k$, its head of queue packet is removed and goes to the server. This is the task of $Dequeue(\cdot)$ function shown in table \ref{table1}. It also checks the queue of user $i^*$ to determine whether it is backlogged. If user $i^*$ is not backlogged anymore, it updates the set of backlogged users.

Upon selection of packet $p$ of user $i^*$ for service, its service tag is incremented by $L_p/\varphi_{i^*}$. Taking this into account, the work level of servers eligible to give service to user $i^*$ should be updated (This is performed by (\ref{eq8_1}) and (\ref{eq8_2}) in $UpdateV(i^*,k)$ subroutine). Therefor, at any time during the execution of the algorithm, $V^k$ represents work progress from server's $k$ perspective.


To understand the en-queuing mechanism (En-queuing module in the algorithm of Table \ref{table1}), assume that user $i$ becomes backlogged after a period of idleness. In order to achieve the memoryless behaviour, user $i$ should not claim its share of service of the idle period from none of the servers. Therefor, $F_i$ is updated to $\max\{F_i,\max\{V^k\mid\pi_{i,k}=1\}\}$, which is greater than or equal to $V^k$, for any server $k$ that is eligible to give service to user $i$, (\eqref{eq4} in Table \ref{table1}).


To provide a better understanding of the algorithm, we consider some illustrative examples. In these examples, we assume that $UpdateV(i^*,k)$ subroutine is restricted to updates in (\ref{eq8_1}) and (\ref{eq8_2}) and other updates in it are not implemented. The resultant algorithm will be referred to as \emph{reduced $CM^4FQ$} algorithm.

The first example describes the role of $ActivateServers(\cdot)$ subroutine. The second example examines the role of \eqref{eq4} by which $F_i$ is updated for a returning customer $i$. The third example, reveals the shortage of the reduced $CM^4fQ$ algorithms to achieve fairness and shows the necessity for additional updates (i.e. \eqref{eq9_0}-\eqref{eq9_4}) in $UpdateV(i,k)$ subroutine of Table \ref{table2}.

\begin{example}\label{exmp3}
Consider again the queuing system shown in Fig. \ref{fig2}. Assume that input traffic stream of users comprises of packets of fixed length $L=1000$ bits. User $c$ is assumed to be continuously backlogged since $t=0$.
It is assumed that 13 packets already exist in the queue of each user $a$ and $b$ at time $t=0$. It is assumed that no packet arrives at these queues during $[0,t_1)$, till $t_1=20$ msec. Users $a$ and $b$ are assumed to become backlogged at $t_1^+$ and remains backlogged afterwards.

Applying the reduced $CM^4FQ$ algorithm, it can be observed that all existing packets in the queue of users $a$ and $b$ at $t=0$ are given service during interval $[0,t_0)$, where $t_0=10$ msec. In fact, the last existing packet in the queue of user $b$ is chosen by server 1 at $9.6$ msec. Hence, $F_a$ and $F_b$ remain constant during $(9.6,20)$ msec. On the other hand, $F_c$ increases with an average rate of 1 Mbps. Hence, $F_c$ is greater than $F_a$ and $F_b$ at time $t_1$ (as shown in Fig. \ref{fig_exmp3}).

At time 9.6 msec that the set of backlogged users connected to server 1, $B^1$ becomes empty, $V^1$ is set to $\infty$ (as stated in \ref{eq8_1}). The role of $ActivateServers(\cdot)$ subroutine is to reset the work level of this server to a proper finite value as soon as it becomes backlogged. Given that user $a$ and user $b$ become backlogged at $t_1^+$, $V^1$ will be updated to $V^2+\delta$ (\eqref{eq3_0} and (\ref{eq3_1}) in $ActivateServers(\cdot)$ subroutine), where $\delta$ is assumed to be $\delta=2L$.

It can be observed that by updating $V^1$ in this way and by choosing $\delta=2L$, no packet of user $a$ or user $b$ is serviced on server 2 from time $t_1$ onwards (Fig. \ref{fig_exmp3}). In contrast, if such an update were not performed or if $\delta$ were chosen equal to $L^{max}$ (or some smaller amount), it would be possible for users $a$ and $b$ to have some packets serviced on server 2.

\begin{figure}
\centering
\includegraphics[width=3.3in]{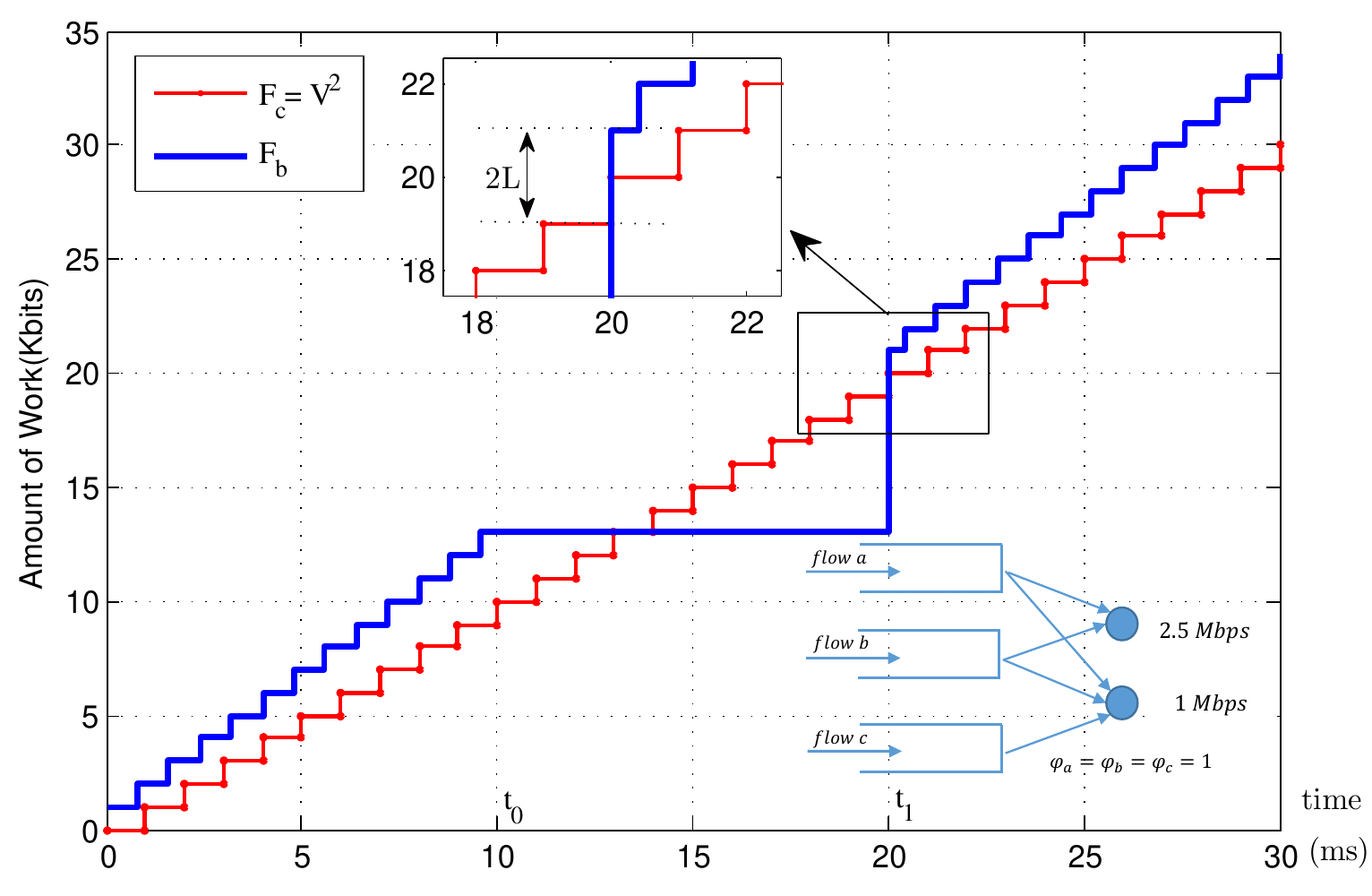}
\caption{Service tag of users $b$ and $c$, $F_b(t)$ and $F_c(t)$ in the network of Fig. \ref{fig2}. Except the interval $(9.6,20]$ msec, users $a$ and $b$ are backlogged elsewhere. Users $c$, is continuously backlogged since $t=0$ (Example \ref{exmp3}).}
\label{fig_exmp3}
\end{figure}
\end{example}

\begin{example}\label{exmp3.1}
Consider again the queuing system shown in Fig. \ref{fig2}. For simplicity assume $L^{max}$, the maximum length of packets, to be infinitesimal (i.e., we consider fluid flow arrivals). Let users $b$ and $c$ be continuously backlogged since time $t=0$. User $a$ becomes backlogged at $t_0 = 100$ msec, and remains backlogged afterwards.
Applying the reduced $CM^4FQ$ algorithm, during interval $(0,t_0)$ user $b$ gets a service rate of 2.5 Mbps from server 1 and user $c$ gets a service rate of 1 Mbps from server 2. Work level of servers are shown in Fig. \ref{fig_exmp1} as function of time.

When user $a$ becomes backlogged at $t_0$, $F_a(t_0^+)$ will be updated from 0 to $\max\{V^k(t_0)\mid\pi_{a,k}=1\}$, (as stated in \eqref{eq4} and shown in Fig. \ref{fig_exmp1}). This way of updating has some consequences. First, $F_a(t_0^+)$ becomes greater than or equal to $V^k(t_0)$, for any server $k$ that is eligible to give service to user $a$. Hence, user $a$ does not claim its share of service of the idle period from non of the servers (memoryless property).

Secondly, user $a$ joins to the cluster with the maximum service rate available to it, i.e., it only gets service from servers 1. Sharing the total service rate of servers 1 between users $a$, $b$, each of them gets a service rate of 1.25 Mbps. On the other hand, user $c$ gets 1 Mbps from server 2. It could be observed that by applying the reduced $CM^4FQ$ algorithm in this case, $CM^4$ fair service rates have been provided for all users.
\end{example}

\begin{figure}
\centering
\includegraphics[width=3.3in]{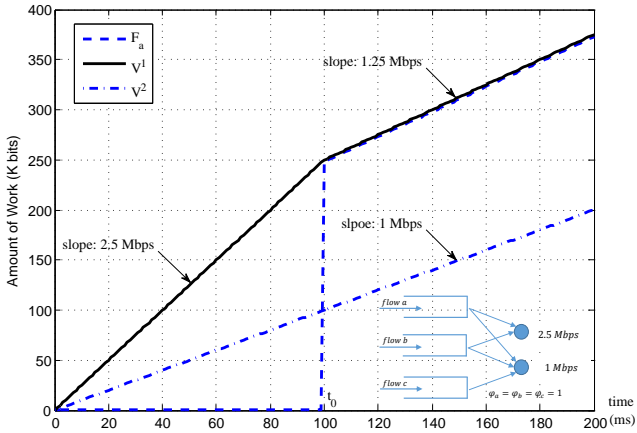}
\caption{Work level of servers, $V^k(t)$ $(k=1,2)$ and service tag of user $a$, $F_a(t)$ in the network of Fig. \ref{fig2} (Example \ref{exmp3.1}).}
\label{fig_exmp1}
\end{figure}

In this example, when user $a$ becomes backlogged, it joins to the cluster with the maximum service rate available to it. As a result, the service rate of this cluster is reduced from 2.5 to 1.25 Mbps (Fig. \ref{fig_exmp1}). However, the configuration of clusters (in terms of servers) is not changed. In general, when some users change their states from backlogged to non-backlogged (or vice versa), the configuration of clusters could be changed entirely. The next example considers such a case.

\begin{example}\label{exmp2}
Consider the queuing system shown in Fig. \ref{fig3-1}. Let users $a$, $b$ and $c$ be continuously backlogged since $t=0$. User $d$ is assumed to become backlogged at $t_0$ and remains backlogged afterwards. For the sake of illustration, suppose that the maximum length of packets is infinitesimally small ($L^{max}\rightarrow0$).

During interval $(0,t_0)$ that user $d$ is not backlogged, the corresponding FOC consists of three clusters: $C_1=(\{a,b\},\{S_1\})$, $C_2=(\{c\},\{S_2\})$ and $C_3=(\{d\},\{\})$, where $r^{C_1}=1.25$ Mbps, $r^{C_2}=1$ Mbps and $r^{C_3}=0$. After $t_0$ that user $d$ becomes backlogged, the corresponding FOC consists of only one cluster which contains all users and all servers, $C=(\{a,b,c,d\},\{S_1,S_2\})$, where $r^C=0.875$ Mbps. Hence, the fact that user $d$ becomes backlogged results in an entirely new FOC.

Furthermore, by applying the reduced $CM^4FQ$ algorithm, it can be observed that during interval $(0,t_0)$ each user achieves its fair service rate. Specifically, $F_a$ and $F_b$ increase with a rate of $1.25$ Mbps, while user $c$ gets a service rate of $1$ Mbps. Consequently, the rate of increase in $V^1$ is $1.25$ times the rate of increase in $V^2$ (Fig. \ref{fig_exmp2}). Hence, as time passes by, a significant difference may form between $V^1$ and $V^2$. Such a gap may potentially disturb fair performance of the reduced $CM^4FQ$ algorithm when the set of backlogged users changes.

When user $d$ becomes backlogged at time $t_0=100$ msec, the fair service rate of each user will be $0.875$ Mbps, i.e., the total capacity should be equally divided among the four users. Nevertheless, as long as $V^1$ is greater than $V^2$ ($F_a$ and $F_b$ are greater than $F_c$), users $a$ or $b$ are not chosen by server $2$ according to selection criterion in \eqref{eq5}. Hence, users $a$ and $b$ get service only from server 1 during $(t_0,t_1)$, till $t_1=250\text{ }$msec. In this interval, users $a$, $b$ and $d$ share server 1 getting $0.833$ Mbps each, while user $c$ gets a service rate of $1$ Mbps. From time $t_1$ when $V^1(t_1)=V^2(t_1)$ (i.e., $V^1(t)$ and $V^2(t)$ cross), each user will get its fair rate of $0.875$ Mbps, (Fig. \ref{fig_exmp2}).

\begin{figure}
\centering
\includegraphics[width=2.5in]{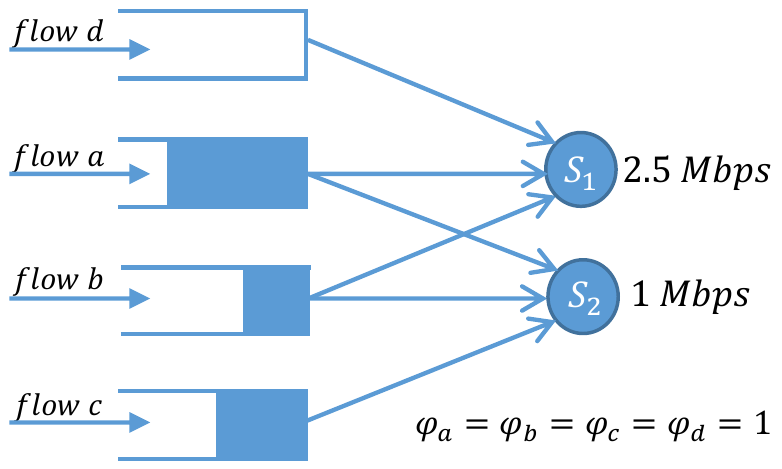}
\caption{The queuing system of Example \ref{exmp2}. Users $a$, $b$ and $c$ are continuously backlogged since $t=0$. User $d$ becomes backlogged at time $t_0=100 msec$, and remains backlogged afterwards.}
\label{fig3-1}
\end{figure}

\begin{figure}
\centering
\includegraphics[width=3.3in]{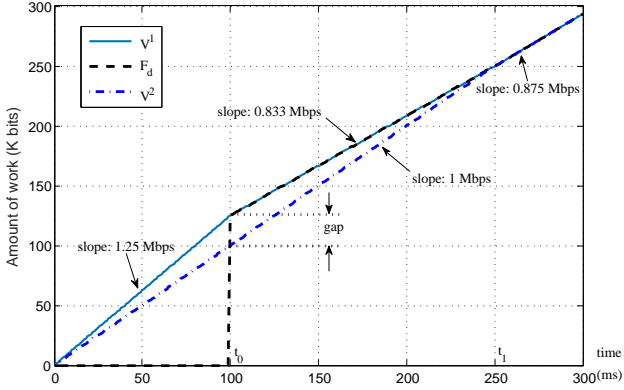}
\caption{Work level of servers, $V^1(t)$ and $V^2(t)$, and service tag of user $d$, $F_d(t)$ in the network of Fig. \ref{fig3-1}, when applying the reduced $CM^4FQ$ algorithm (Example \ref{exmp2}). 
}
\label{fig_exmp2}
\end{figure}
\end{example}
Example \ref{exmp2} shows that when the set of backlogged users changes, the gaps among work level of servers (as computed by the reduced $CM^4FQ$ algorithm) could disturb fair performance of the reduced $CM^4FQ$ algorithm during a considerable amount of time. In order to avoid this situation, for any server $k$ the difference of $V^{k}$ and $\max\{V^l\mid V^l<V^{k}\}$ is limited in $CM^4FQ$ algorithm to the predetermined parameter $\delta$. Specifically, when this difference is exceeding $\delta$ (or equivalently $d^k$ that is defined in (\ref{eq9_0}) becomes positive), the service tag of a subset of users, which have had higher service rate than the other users, is decreased according to (\ref{eq9_1}). Also, the work level of each server is correspondingly matched to the updated service tags (as stated in \eqref{eq9_3}). To get more intuition, consider the following example.

\begin{example}\label{exmp2.1}
Consider again the queuing system described in Example \ref{exmp2}. Assume that input traffic stream of each user comprises of packets of fixed length $L=1000$ bits. Users $a$, $b$ and $c$ are continuously backlogged since $t=0$. User $d$ is not backlogged during interval $[0,40]$ msec.

Applying $CM^4FQ$ algorithm, the work level of servers are shown in Fig \ref{fig_exmp2.1}. In this example, whenever $V^1$ is going to exceed $V^2+\delta$, $F_a$, $F_b$ and accordingly $V_1$ are decreased (as stated in \eqref{eq9_1} and \eqref{eq9_3}), so that $|V^1-V^2|$ is always upper bounded by $\delta$, where $\delta$ is assumed to be $\delta=2L$.
\end{example}
\begin{figure}
\centering
\includegraphics[width=3.3in]{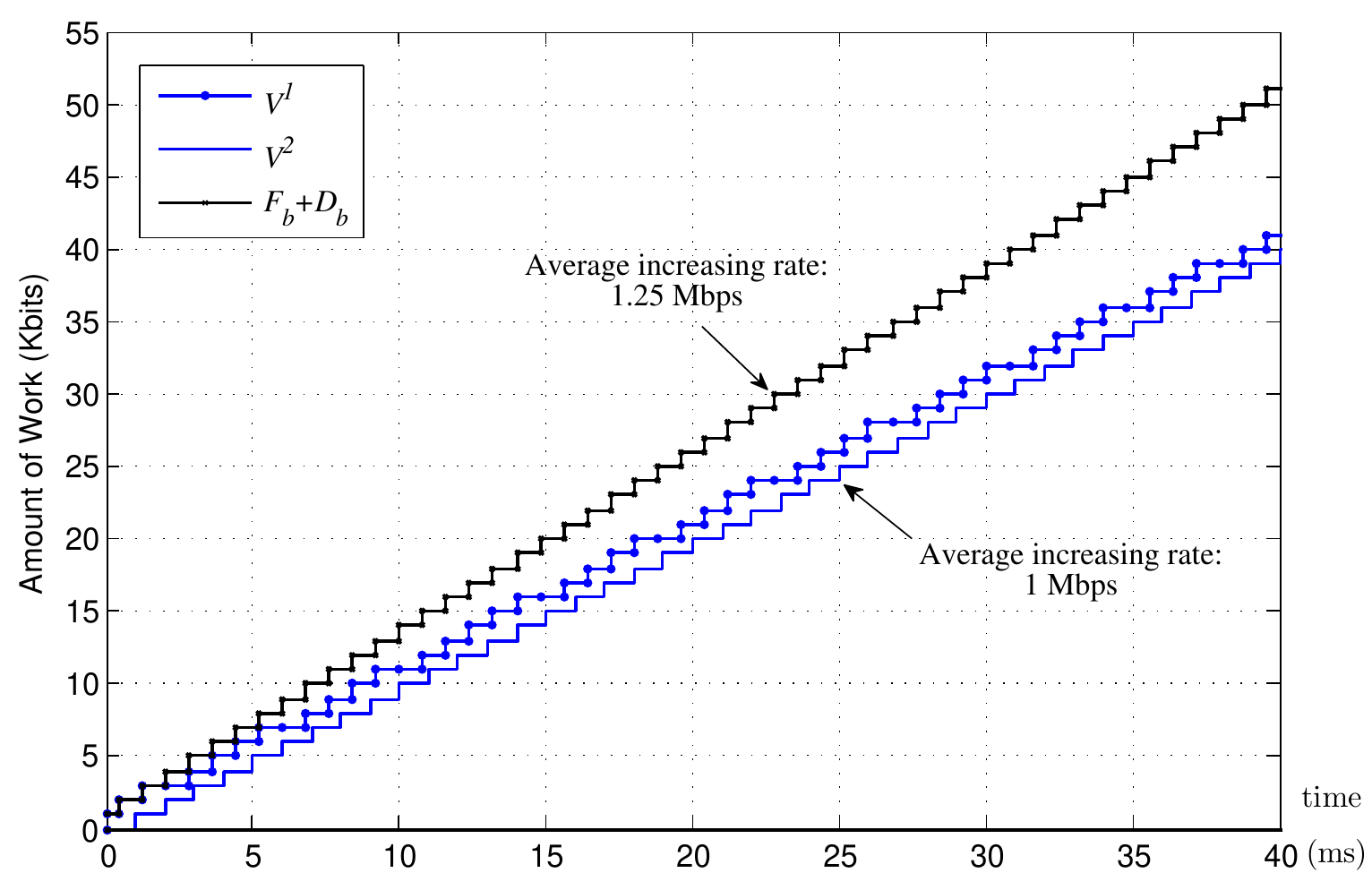}
\caption{Work level of servers, $V^1(t)$ and $V^2(t)$, and the given work to user $b$, $F_b(t)+D_b(t)$ in the network of Fig. \ref{fig3-1}, when applying $CM^4FQ$ algorithm (Example \ref{exmp2.1}).}
\label{fig_exmp2.1}
\end{figure}

As Example \ref{exmp2} implies, it is desired to choose $\delta$ as small as possible. However, in order to achieve fairness, $\delta$ should be sufficiently large. For instance, if we choose $\delta=L$ in Example \ref{exmp2.1}, it will be possible for $V^1(t)$ to be equal to $V^2(t)$ at some instants. In this case, it will be possible for users $a$ and $b$ to be chosen by server 2 during interval $(0,t_0)$. However, by choosing $\delta=2L$, no packet of user $a$ or user $b$ is chosen by server 2 during $(0,t_0)$. Choosing a proper amount for $\delta$ will be discussed in details in the next section.

It can be concluded from Example \ref{exmp2.1} that $F_i(t)$ does not show the actual work which is allocated to user $i$. To calculate the actual work allocated to user $i$, we calculate $D_i(t)$ which represents the amount of work that is allocated to user $i$ (during $[0,t)$) but is not considered in $F_i(t)$. Specifically, whenever $F_i$ is decreased due to updates in \eqref{eq9_1}, $D_i$ is increased by the same amount (as stated in \eqref{eq9_2}). $D_i(t)$ is referred as the \emph{bonus} given to user $i$ during $[0,t)$. As an illustrative instance, we have calculated $D_b(t)$ in Example \ref{exmp2} and we have shown the actual work allocated to user $b$ during $[0,t)$, $F_b(t)+D_b(t)$, in Fig. \ref{fig_exmp2.1}.

In parallel with $D_i(t)$, a similar parameter, $D^l(t)$ is calculated per server. Calculating $D_i$ and $D^l$ parameters in $UpdateV(\cdot,\cdot)$ subroutine is not an integral part of the algorithm. However, they will be useful for analyzing the actual work that is given to users.


Now that we have a whole picture of $CM^4FQ$ algorithm, we compare it with single server Start-time Fair Queuing (SFQ) algorithm \cite{SFQ}. If we restrict number of servers in $CM^4FQ$ algorithm to one server, it will be almost the same as single server SFQ algorithm. The only difference is in definition of work level function. Specifically, in SFQ when the server chooses user $i^*$ for service, work level of server is simply updated to ${F}_{i^*}$, where ${F}_{i^*}$ is the service tag of user $i^*$ at the time that it is chosen for service.

The intuition behind the definition which we proposed for work level function in $CM^4FQ$ algorithm reveals in the next example. An analysis of $CM^4FQ$ algorithm under single server assumption is expressed in the appendix.


\begin{example}\label{exmp4}
Consider the queuing system shown in Fig. \ref{fig_exmp40}. Assume that the traffic stream of each user comprises of packets of fixed length $L=1000$ bits. Both users are assumed to be continuously backlogged since $t=0$. It can be easily verified that each user has a fair service rate of 5.5 Mbps. In order to achieve fairness, assume $\delta$ to be sufficiently large that the condition ($d^k>0$) in UpdateV$(i,k)$ subroutine never holds.

Applying $CM^4FQ$ algorithm for scheduling packets on servers, the work level of each server is shown in Fig. \ref{fig_exmp4} as function of time. It is observed that the maximum difference between $V^1$ and $V^2$ in $CM^4FQ$ algorithm is limited to $2L$ at any time. 
\begin{figure}
\centering
\includegraphics[width=2.5in]{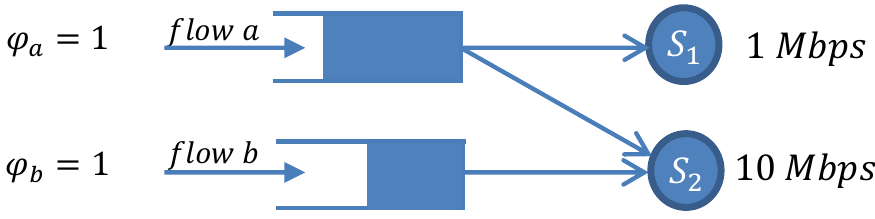}
\caption{The queuing system of Example \ref{exmp4}. Users are assumed to be continuously backlogged since $t=0$.}
\label{fig_exmp40}
\end{figure}

\begin{figure}
\centering
\includegraphics[width=3.3in]{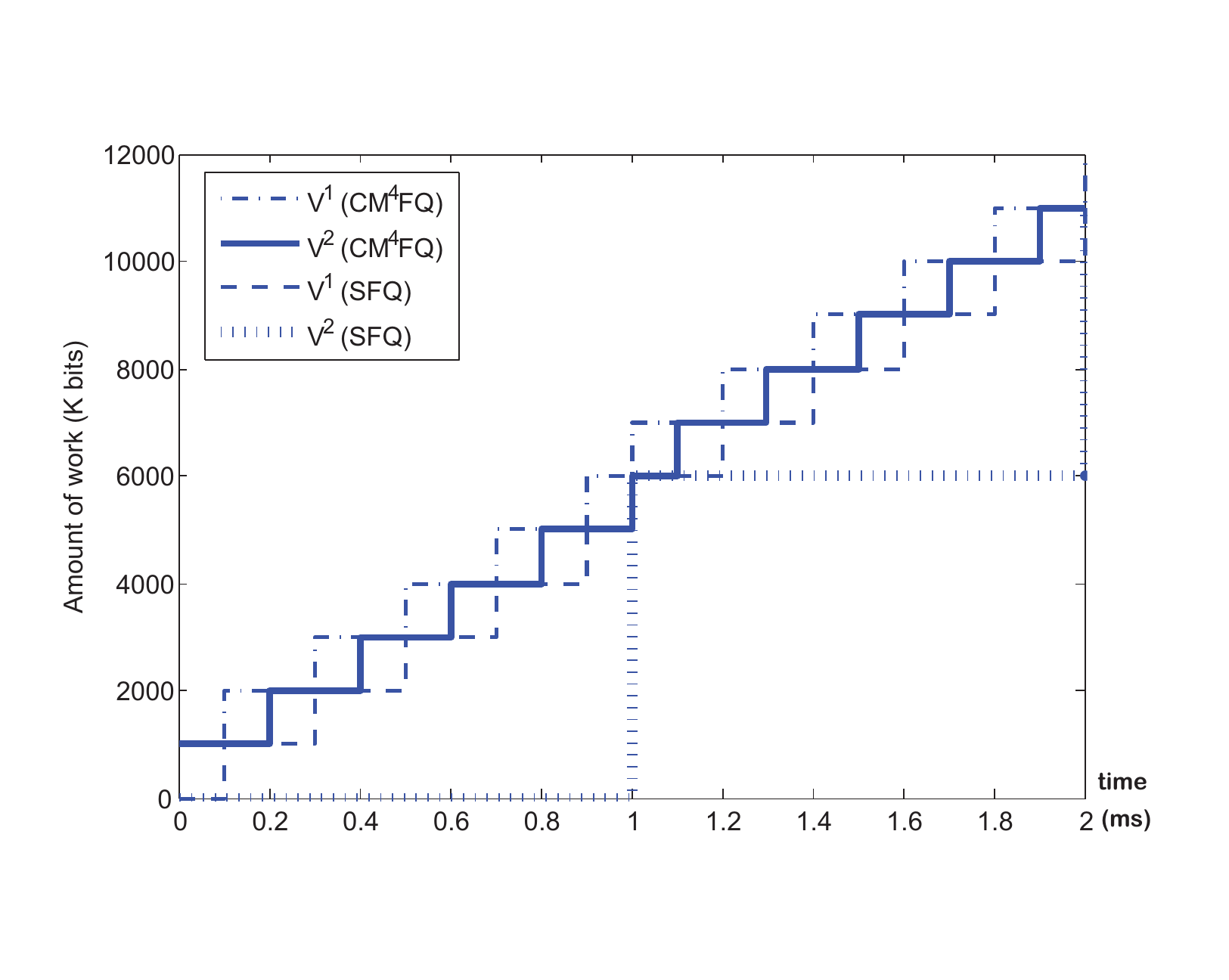}
\caption{Comparison between work level of servers in $CM^4FQ$ algorithm and the SFQ-based algorithm (Example \ref{exmp4}).}
\label{fig_exmp4}
\end{figure}

Now consider a FQ algorithm which is the same as $CM^4FQ$ algorithm except that it sets the work level of each server as in SFQ algorithm. This algorithm will be referred to as SFQ-based algorithm. Applying SFQ-based algorithm, work level of servers, $V^1$(SFQ) and $V^2$(SFQ) are shown in Fig \ref{fig_exmp4}.

In SFQ-based algorithm the work level of each server is not changed during service time of a packet. In this example, during service time of a packet on server $1$, ten packets are given service by server $2$. Hence, $V^1-V^2$ in SFQ-based algorithm experiences a maximum difference of $5L$ (Fig. \ref{fig_exmp4}).$\IEEEQED$
\end{example}

The above example considered a situation where all users have the same fair service rate. As this example implies, depending on the ratios of the service rate of servers, the maximum difference between work level of servers in the SFQ-based algorithm could be so much large. In the next section, we will show that the maximum difference between work level of servers in $CM^4FQ$ algorithm is upper bounded by $(K+1)L$ (K is the total number of servers), provided that all users have the same fair service rate.

\section{Characterizing Properties and Performance of $CM^4FQ$ Algorithm}
$CM^4FQ$ algorithm will be discussed in this section more analytically. First, some analytic properties of $CM^4FQ$ algorithm are described. Then the performance of the algorithm is studied in a special case. Finally, we discuss how to choose the control parameter $\delta$.

\subsection{Properties of $CM^4FQ$ Algorithm}
Basic properties of $CM^4FQ$ algorithm are discussed in this subsection. These properties turn out to be useful in characterizing performance of the algorithm. Wherever necessary, we develop new definitions and notation.

\begin{definition}
For any user $i$, the \emph{respective work level} function, $V_i(t)$ is defined as:
\begin{equation}\label{eq3_2}
  V_i(t)=\max\{V^k(t)\mid\pi_{i,k}=1\}
\end{equation}
\end{definition}

\begin{lem}\label{lem2_1}
For user $i$ that is backlogged at time $t$, it follows that $F_i(t)\ge V_i(t)$.
\end{lem}
\begin{proof}
According to (\ref{eq4}), for an idle user which becomes backlogged at $t$, $F_i(t^+)$ is updated to $\max\{V_i(t),F_i(t)\}$, which is greater than or equal to $V_i(t)$. For user $i$ that is already backlogged at time $t$, equations (\ref{eq4}) and \eqref{eq8_2} implicate that $F_i(t)$ is greater than or equal to $V^k(t)$, for any server $k$ for which $\pi_{i,k}=1$. According to (\ref{eq3_2}), it follows that $F_i(t)\ge V_i(t)$.
\end{proof}

\begin{definition}
The \emph{allocated work} to user $i$ during interval $[t_1,t_2)$, $W_i(t_1,t_2)$ is defined as the total required work of user $i$'s packets \emph{chosen} during interval $[t_1,t_2)$ for getting service.
\end{definition}

As defined in Section III.A, the length of a packet in bits is considered as its required work. It should be mentioned that the service of packets chosen during $[t_1,t_2)$ could possibly be completed after $t_2$. However, the total required work of these packets is considered in $W_i(t_1,t_2)$. On the other hand, there could be some packets of user $i$ which are chosen before $t_1$ and their service complete after $t_1$. These residual works are not considered in $W_i(t_1,t_2)$.

\begin{definition}
The set of users backlogged at some instant(s) during interval $[t_1,t_2)$ is denoted by $B(t_1,t_2)$:
\begin{equation}\label{eq10}
  B(t_1,t_2)\triangleq\{i\mid i\in B(t),\text{ for some }t\in[t_1,t_2)\}.
\end{equation}
The set of users continuously backlogged during interval $[t_1,t_2)$ is denoted by $\tilde{B}(t_1,t_2)$:
\begin{equation}\label{eq11}
  \tilde{B}(t_1,t_2)\triangleq\{i\mid i\in B(t),\text{ for all }t\in[t_1,t_2)\}
\end{equation}
\end{definition}
The following lemma holds for any user that is continuously backlogged during $[t_1,t_2)$.

\begin{lem}
For user $i\in\tilde{B}(t_1,t_2)$, the normalized work during interval $[t_1,t_2)$ equals to:
\begin{equation}\label{eq12}
  \frac{W_i(t_1,t_2)}{\varphi_i}=F_i(t_1,t_2)+D_i(t_1,t_2)
\end{equation}
\end{lem}
\begin{IEEEproof}
The proof is so straightforward. So, for the sake of brevity it is not stated.
\end{IEEEproof}

\subsection{Analyzing $CM^4FQ$ Algorithm in Case of a Single Cluster}
In this section we analyze $CM^4FQ$ algorithm in case of a single cluster. To do so, first we present exact definition of FOC in case of packetized traffic streams. 
According to Fact \ref{cor1}, FOC in a fluid flow system has an important property. Specifically, in a fluid flow system, FOC at time $t_0$ is expressed as a function of the set of backlogged users at $t_0$, $\{C_m(B_0)\}_{m=1}^M$, where $B_0=B(t_0)$. Upon this property, FOC definition could be generalized for queuing systems with packetized traffic streams.

\begin{definition}
Assume that in the queuing system of Fig. \ref{fig1} users are given service on a packet by packet basis. Given that the set of backlogged users at time $t_0$ equals to $B_0=B(t_0)$, FOC at time $t_0$ is defined to be the same as FOC in a fluid flow system with the set of backlogged users $B_0$, $\{C_m(B_0)\}_{m=1}^M$.
\end{definition}

\begin{definition}
Assume that the set of backlogged users at time $t_0$ is $B_0=B(t_0)$. Consider an arbitrary cluster $C=(I,S)$ of the FOC which corresponds to $B_0$, $\{C_m(B_0)\}_{m=1}^M$. \emph{Minimum work level} function ($V^C(t_0)$), and \emph{maximum work level} function ($\hat{V}^C(t_0)$) corresponding to cluster $C$ are defined as:
\begin{eqnarray}
  V^C(t_0) &\triangleq& \min\{V^k(t_0)\mid k\in S\}\label{eq13} \\
  \hat{V}^C(t_0) &\triangleq& \max\{V^k(t_0)\mid k\in S\}\label{eq14}
\end{eqnarray}
\end{definition}

An important question is that how far can $\hat{V}^C(t)$ be from $V^C(t)$ for an arbitrary cluster $C$. To find an answer to this question, we will consider a special situation. Specifically, consider an arbitrary interval $[{t}_0,t_1)$ during which the set of backlogged users does not change, i.e., $\tilde{B}({t}_0,t_1)=B({t}_0,t_1)=B$. Consider an arbitrary cluster $C=(I,S)$ of the FOC corresponding to the set of backlogged users $B$. Assume that an initial condition, $V^k({t}_0)=V^{C}({t}_0)$ holds for all servers $k\in S$. Furthermore, assume that in the execution of the algorithm during interval $[t_0,t_1)$ cluster $C$ is \emph{isolated} from the rest of network, i.e., during this interval no packet of users $i\notin I$ is chosen for service by servers $k\in S$, and no packet of users $i\in I$ is chosen for service by servers $k\notin S$. The following results are derived under such assumptions.

\begin{theorem}\label{th2}
Consider an arbitrary interval $[{t}_0,t_1)$, for which $\tilde{B}({t}_0,t_1)=B({t}_0,t_1)=B$. Consider an arbitrary cluster $C=(I,S)$ of the FOC which corresponds to the set of backlogged users $B$. Assume that in the execution of the algorithm during interval $[t_0,t_1)$ cluster $C$ is {isolated} from the rest of network. Let $V^k({t}_0)=V^{C}({t}_0)$, for all servers $k\in S$.

Assume that $\delta$ is chosen arbitrarily large that $D^k({\tau}_0,\tau_1)=D^{l}({\tau}_0,\tau_1)$ for any servers $k,l\in S$ and for any subinterval $[{\tau}_0,\tau_1)$ of the interval $[t_0,t_1)$. Given that server $k\in S$ becomes free at $t\in(t_0,t_1)$, it can be shown that:
\begin{eqnarray}\label{eq17}
  {w}^{k}(t_0,t)\leq r^{C}(t-t_0)+K^{C}\lambda^{C}
\end{eqnarray}
where, ${w}^{k}(t_0,t)$ is defined as:
\begin{eqnarray}
  {w}^{k}(t_0,t)\triangleq {V}^{k}(t_0,t)+{D}^{k}(t_0,t),
\end{eqnarray}
$K^{C}$ is the number of servers in cluster $C$, and $\lambda^{C}$ is defined as:
\begin{eqnarray}
  \lambda^{C} &\triangleq& \frac{L^{max}}{\min_{i\in I}\varphi_i}\label{eq18}.
\end{eqnarray}

Furthermore, for server $k\in S$ and for any $t\in(t_0,t_1)$, it can be shown that:
\begin{eqnarray}\label{eq23}
 w^{k}(t_0,t)\ge r^{C}(t-t_0)-K^{C}\lambda^{C}.
\end{eqnarray}
\end{theorem}
For the proof, refer to the appendix.
%
%
%

\begin{theorem}\label{th2'}
Consider the same set of conditions as in Theorem \ref{th2}. Given that server $k\in S$ becomes free at $t\in(t_0,t_1)$, it follows that:
\begin{equation}\label{eq22}
  {V}^{k}(t^+)-V^{C}(t)\leq (K^C+1)\lambda^C.
\end{equation}
\end{theorem}

For the proof, refer to the appendix.
Intuitively, Theorem \ref{th2'} implies that the maximum difference among work level of servers in cluster $C$ is limited to $(K^C+1)\lambda^C$. The following lemma is built upon this theorem.

\begin{lem}\label{lem2}
For $D^k(\tau_0,\tau_1)=D^{l}(\tau_0,\tau_1)$, $\forall k,l\in S$, to be established under conditions of Theorem \ref{th2}, it suffices to choose $\delta\ge (K^C+1)\lambda^{C}$.
\end{lem}

\begin{IEEEproof}
Consider any time $t\in(t_0,t_1)$ at which some user (server) in cluster $C$ receives bonus. Specifically,
Consider an arbitrary time $t\in(t_0,t_1)$ at which some server $k\in S$ becomes free. At this instant, some user $i^*$ will be chosen according to \eqref{eq5} and $UpdateV(i^*,k)$ subroutine will be executed.

According to \eqref{eq22}, unless server $k$ has the minimum work level at time $t$, there exists at least one server $l\in S$ such that $V^k(t^+)-V^l(t)\le(K^C+1)\lambda^C$. Hence, $d^k$ in \eqref{eq9_0} will be less than or equal to 0, provided that $V^k(t)>V^C(t)$. In this case no user (server) in cluster $C$ receives bonus. $d^k$ is positive only when $V^k(t)=V^C(t)$. In this case, all users (servers) in cluster $C$ receive the same amount of bonus (because $F_i(t)$ and $V^l(t)$ is greater than or equal to $V^C(t)$ for all users and all servers in cluster $C$). Hence, for any subinterval $[\tau_0,\tau_1)\subseteq[t_0,t_1)$ it follows that $D^k(\tau_0,\tau_1)=D^{l}(\tau_0,\tau_1)$, $\forall k,l\in S$.
\end{IEEEproof}

\begin{corollary}\label{cor1.2}
Consider an interval $[{t}_0,t_1)$ for which $\tilde{B}({t}_0,t_1)=B({t}_0,t_1)=B$. Consider a special case that the FOC corresponding to $B$ include one cluster $C=(I,S)$ which consists of all backlogged users and all $K$ servers.
Let $V^k({t}_0)=V^{C}({t}_0)$, for all servers. In order for any user $i\in I$ to get the same amount of bonus during interval $[t_0,t_1)$, it suffices to choose $\delta\ge (K+1)\lambda_0$, where $\lambda_0$ is defined as:
\begin{equation}\label{eq27}
  \lambda_0\triangleq\frac{L^{max}}{\min_{\forall i}\varphi_i}.
\end{equation}
\end{corollary}

By choosing $\delta\ge (K+1)\lambda_0$, we can make sure that under conditions of Theorem \ref{th2} and for any possible configuration of the FOC and in any cluster $C=(I,S)$, all users receive the same amount of bonus.

\subsection{Choosing the Control Parameter $\delta$}
In this section, we discuss how to choose the parameter $\delta$. As it stems from Example \ref{exmp2}, it is desired to choose $\delta$ as small as possible. On the other hand, according to Corollary \ref{cor1.2} it is desired to choose $\delta\ge (K+1)\lambda_0$. Hence, we propose to choose $\delta=(K+1)\lambda_0$.

\emph{Proposition:} The control parameter, $\delta$ in $CM^4FQ$ algorithm could be chosen as:
\begin{equation}\label{eq26}
  \delta=(K+1)\lambda_0
\end{equation}
where $\lambda_0$ is defined in \eqref{eq27}.
In the rest of this paper, $\delta$ is chosen according to (\ref{eq26}). The following Theorem gives more insights on the reason of this choice.

\begin{theorem}\label{lem6}
Consider an arbitrary interval $[t_0,t_1)$ for which ${B}(t_0,t_1)={B}$. Assume that cluster $C_0=(I_0,S_0)$ in the FOC corresponding to $B$, $C_0\in\{C_m({B})\}_{m=1}^M$ has the minimum service rate. In addition, assume that any user $j\in I_0$ is continuously backlogged during interval $[t_0,t_1)$.

Assume there exists a time instant $\hat{t}_0\in [t_0,t_1)$ at which $V^{C_m}(\hat{t}_0^+)-\hat{V}^{C_0}(\hat{t}_0^+)\geq\delta$, for all clusters $C_m, m>0$. Applying $CM^4FQ$ algorithm for scheduling packets on servers, it can be shown that during $(\hat{t}_0,t_1)$ no packet of users $j\in I_m, m>0$ is chosen for service by servers $k\in S_0$.
\end{theorem}

Theorem \ref{lem6} considers an arbitrary interval for which the steady-state condition does not necessarily hold. To prove this Theorem, we need the following Lemmas which generalize the results of Theorem \ref{th2} for non steady state intervals. The proof of Theorem \ref{lem6} is provided in the appendix.

\begin{lem}\label{lem4}
Consider an arbitrary interval $[t_0,t_1)$ for which $B(t_0,t_1)=B$. Assume that cluster $C_M=(I_M,S_M)$ in the FOC corresponding to $B$, $C_M\in\{C_m(B)\}_{m=1}^M$ has the maximum service rate. Let $V^{C_M}(t)<\infty$ at any instant $t\in[t_0,t_1)$. Regarding interval $[t_0,t)$, $t_0<t\le t_1$, it can be shown that:
\begin{equation}\label{eq28}
  V^{C_M}(t_0,t)+D^{C_M}(t_0,t)\geq r^{C_M}(t-t_0)-K^{C_M}\lambda^{C_M}
\end{equation}
where $D^{C}(t_0,t)$ is defined for any cluster $C=(I,S)$ as:
\begin{equation}\label{}
  D^{C}(t_0,t)\triangleq\min_{k\in S}D^k(t_0,t).
\end{equation}
\end{lem}

\begin{IEEEproof}
To derive a lower bound on $V^{C_M}(t_0,t)+D^{C_M}(t_0,t)$, consider the worst case that any user $j\in I_M$ is continuously backlogged during $[t_0,t)$, and for any user $j\in I_M$, $F_j(t_0)$ equals to $V^{C_M}(t_0)$. Furthermore, assume that during interval $[t_0,t_1)$ no packet of users $j\in I_M$ is given service by servers $k\notin S_M$. When considering such worst case conditions (which result in the minimum of $V^{C_M}(t_0,t)+D^{C_M}(t_0,t)$), cluster $C_M$ could be treated as an isolated cluster for which all conditions of Theorem \ref{th2} hold in conjunction with the interval $[t_0,t)$. Hence, (\ref{eq23}) holds for \emph{all servers} $k\in S_M$. This consequently results in \eqref{eq28}.
\end{IEEEproof}

\begin{lem}\label{lem5}
Consider an arbitrary interval $[t_0,t_1)$ for which $\tilde{B}(t_0,t_1)=\tilde{B}$. Assume that cluster $C_0$ in the FOC corresponding to $B$, $C_0\in\{C_m(\tilde{B})\}_{m=1}^M$ has the minimum service rate. Given that server $k\in S_0$ becomes free at time $t\in(t_0,t_1)$, it can be shown that:
\begin{eqnarray}\label{eq29}
  {V}^{k}(t)-\hat{V}^{C_0}(t_0)+D^{k}(t_0,t)\qquad\qquad\quad\:\\\nonumber
  \leq r^{C_0}(t-t_0)+K^{C_0}\lambda^{C_0}
\end{eqnarray}
\end{lem}

\begin{IEEEproof}
To derive an upper bound on ${V}^{k}(t)-\hat{V}^{C_0}(t_0)+D^{k}(t_0,t)$, consider the best situation that during interval $[t_0,t_1)$ only users $j\in I_0$ are given service by servers $l\in S_0$. Furthermore, assume that $V^l(t_0)=\hat{V}^{C_0}(t_0)$ for all servers $l\in S_0$. When considering the best situation (which results in the maximum of ${V}^{k}(t)-\hat{V}^{C_0}(t_0)+D^{k}(t_0,t)$), cluster $C_0$ could be treated as an isolated cluster for which all conditions of Theorem \ref{th2} hold in conjunction with the interval $[t_0,t)$. Therefor, (\ref{eq17}) holds for server $k$ which becomes free at $t$. Consequently, (\ref{eq29}) is deduced from (\ref{eq17}).
\end{IEEEproof}

Upon the results developed in this section, the throughput performance of $CM^4FQ$ algorithm will be analyzed in the next section.

\section{Performance Analysis}
In the context of single server FQ, it is known that fairness can not be fully achieved through packet by packet scheduling \cite{SCFQ}. Therefor, some measures of fairness have been developed which are used as benchmarks to study fair performance of FQ algorithms.

Among fairness measures, the most evident is \emph{steady state throughput guarantee}. Specifically, given that the set of backlogged users does not change during interval $[t_0,t_1)$, the work offered to any backlogged user $i$, $W_i(t_0,t_1)$, should have a bounded difference with respect to the user's fair service share. \cite{SFQ}

Regarding $CM^4$ fairness definition, we propose a measure of fairness that is referred as \emph{worst case throughput guarantee}. Subsequently, it will be shown that $CM^4FQ$ algorithm achieves $CM^4$ fairness in the sense of worst case throughput guarantee.
\begin{definition}
Consider an arbitrary interval $[t_0,t_1)$ for which $B(t_0,t_1)=B$ and $\tilde{B}(t_0,t_1)=\tilde{B}$. It is said that an FQ algorithm achieves $CM^4$ fairness in the sense of \emph{worst case throughput guarantee}, when for any user $i$ that is continuously backlogged during $[t_0,t_1)$:
\begin{equation}\label{eq35}
  \frac{r_i(B).[t_1-t_0]}{\varphi_i}-\Delta\leq\frac{W_i(t_0,t_1)}{\varphi_i}\leq\frac{r_i(\tilde{B}).[t_1-t_0]}{\varphi_i}+\tilde{\Delta}
\end{equation}

where, $r_i(B)$ is the fair service rate of user $i$ in the FOC corresponding to $B$, $\{C_m(B)\}$, and $r_i(\tilde{B})$ is the fair service rate of user $i$ in the FOC corresponding to $\tilde{B}$, $\{C_m(\tilde{B})\}$.
\end{definition}

Intuitively, user $i\in \tilde{B}(t_0,t_1)$ receives the minimum amount of service when all users $j\in B(t_0,t_1)$ are almost backlogged during $[t_0,t_1)$. On the other hand, user $i$ receives the maximum amount of service when users $j\in B(t_0,t_1)-\tilde{B}(t_0,t_1)$ are almost idle during $[t_0,t_1)$. Hence, the minimum service rate achievable by user $i$ is $r_i(B)$, which corresponds to the case that all users $j\in B(t_0,t_1)$ are continuously backlogged during $[t_0,t_1)$. On the other hand, the maximum service rate achievable by user $i$ is $r_i(\tilde{B})$, which corresponds to the case that all users $j\in B(t_0,t_1)-\tilde{B}(t_0,t_1)$ are continuously idle during $[t_0,t_1)$.

\begin{theorem}\label{th4}
Consider an interval $[t_0,t_1)$ for which $B(t_0,t_1)=B$. In the FOC corresponding to $B$, let clusters be indexed in an increasing order of their service rates, $\{C_1,...,C_M\}$.
Consider user $i$ that is continuously backlogged during $[t_0,t_1)$ and belongs to some cluster $C_m$.
It can be shown that:
\begin{equation}\label{eq36}
  \frac{W_i(t_0,t_1)}{\varphi_i}\geq \frac{r_i(B).[t_1-t_0]}{\varphi_i}-\Delta^{init}_{i}(t_0)-\lambda_0\sum_{l\ge m}K^{C_l}
\end{equation}
where, $\Delta_{i}^{init}(t_0)$ is defined as:
\begin{equation}\label{eq37}
\Delta_i^{init}(t_0)\triangleq F_i(t_0)-\min\{V^{C_l}(t_0)\mid l\ge m\}
\end{equation}
\end{theorem}
For the proof, refer to the appendix.

Theorem \ref{th4} describes performance of $CM^4FQ$ algorithm in a dynamic situation. Specifically, for an arbitrary interval $[t_0,t_1)$ during which the set of backlogged users possibly changes several times, (\ref{eq36}) characterizes the minimum throughput which could be guaranteed for a backlogged user $i$. The next theorem gives the maximum throughput achievable by a backlogged user.

\begin{theorem}
Consider an interval $[t_0,t_1)$ for which $\tilde{B}(t_0,t_1)=\tilde{B}$. In the FOC corresponding to $\tilde{B}$, let clusters be indexed in an increasing order of their service rates, $\{C_1,...,C_{\tilde{M}}\}$. Consider user $i$ which belongs to some cluster $C_{\tilde{m}}$. It can be shown that:
\begin{equation}\label{eq38}
  \frac{W_i(t_0,t_1)}{\varphi_i}\leq \frac{r_i(\tilde{B}).[t_1-t_0]}{\varphi_i}+\tilde{\Delta}_i^{init}(t_0)+\lambda_0\sum_{l\le \tilde{m}}K^{C_l}
\end{equation}
where, $\tilde{\Delta}_i^{init}(t_0)$ is defined as:
\begin{equation}\label{eq39}
\tilde{\Delta}^{init}_i(t_0)\triangleq\max\{\hat{V}^{C_l}(t_0)\mid l\le \tilde{m}\}-F_i(t_0)+L^{max}/\varphi_i
\end{equation}
\end{theorem}

\begin{IEEEproof}
The proof follows the same line of argument as the proof of Theorem \ref{th4}. So, for the sake of brevity it is not stated.
\end{IEEEproof}

\begin{corollary}\emph{(Worst Case Throughput Guarantee)}\label{cor2}
For any user $i$ that is backlogged at $t_0$, it can be verified that:
\begin{eqnarray}
\Delta^{init}_i(t_0) &\leq& (K-1)\delta+L^{max}/\varphi_i \label{eq40} \\
\tilde{\Delta}^{init}_i(t_0) &\leq& (K-1)\delta+L^{max}/\varphi_i \label{eq41}
\end{eqnarray}
As a result, (\ref{eq35}) is established for user $i$ that is continuously backlogged during $[t_0,t_1)$, when opting $\Delta=\tilde{\Delta}=K\delta$.
\end{corollary}

\begin{corollary} \emph{(Steady state throughput guarantee)}
\label{cor3}
Consider an interval $[t_0,t_1)$ for which $\tilde{B}(t_0,t_1)=B(t_0,t_1)=B$. For user $i$ that is continuously backlogged during $[t_0,t_1)$, it follows that:
\begin{eqnarray}\label{eq42}
  \nonumber \left|\frac{W_i(t_0,t_1)}{\varphi_i}-\frac{r_i(B)[t_1-t_0]}{\varphi_i}\right|\leq\qquad\qquad\\
  K\lambda_0+\max\{\Delta^{init}_i(t_0),\tilde{\Delta}^{init}_i(t_0)\}
\end{eqnarray}
\end{corollary}

Corollary \ref{cor3} describes the performance of $CM^4FQ$ algorithm in a steady state situation. Specifically, given that the set of backlogged users does not change during $[t_0,t_1)$, for any user backlogged user $i$, $W_i(t_0,t_1)$ has a bounded difference with respect to its fair service share.

\section{Numerical Results}
Performance of $CM^4FQ$ algorithm is evaluated in this section through a numerical example. Specifically, consider the queuing system shown in Fig \ref{fig4}, where the weight of each user is equal to one. It is assumed that traffic streams of users are i.i.d packet streams with packets length uniformly distributed between $[800,1000]$ bits. Substituting $L^{max}=1000$ in (\ref{eq26}), $\delta$ is chosen equal to $\delta=5000$.\par
\begin{figure}
\centering
\includegraphics[width=2.5in]{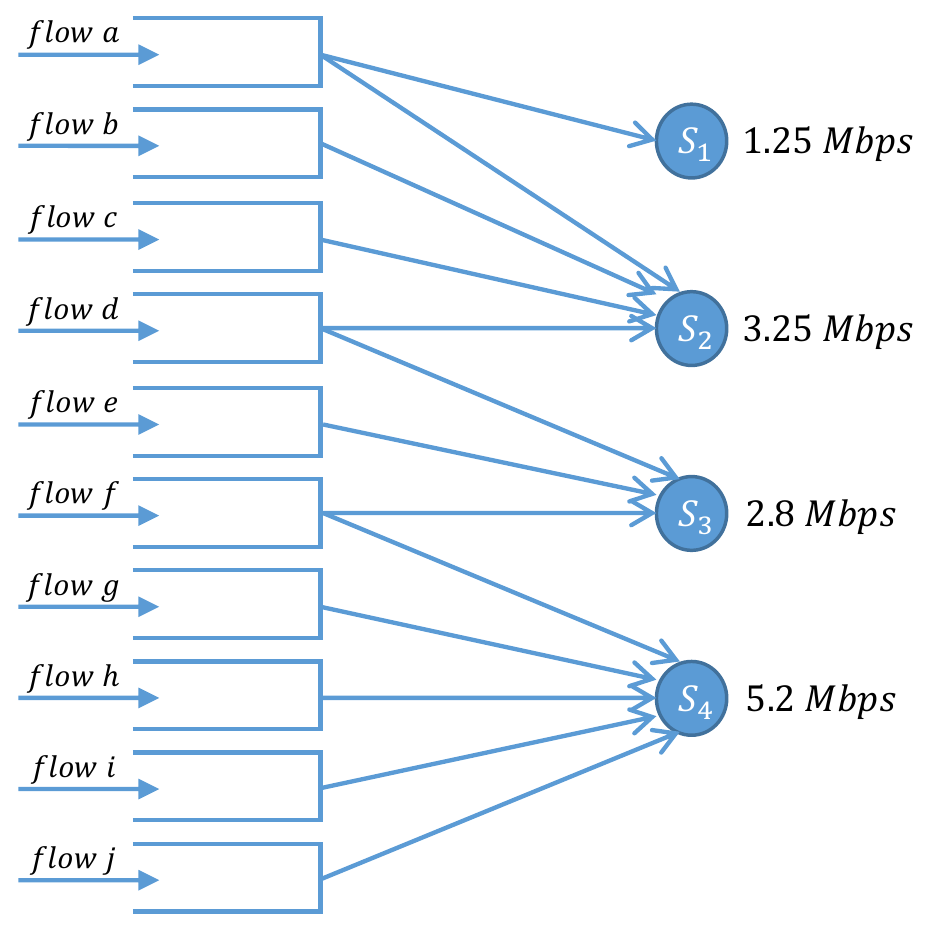}
\caption{The queuing system that is simulated for numerical evaluation of $CM^4FQ$ algorithm.}
\label{fig4}
\end{figure}

As the first study, $CM^4FQ$ algorithm is executed on the network of Fig. \ref{fig4}, while all users are backlogged during interval $[0,20]$ seconds. Assuming all users to be backlogged, in the corresponding FOC all users and all servers form one cluster with the service rate of $1.25$ Mbps. As a result, $\max_{k,l}\{V^k(t)-V^l(t)\}$ is expected to have an upper bound of $5000$ bits as stated in (\ref{eq22}). The maximum difference between work level of servers is shown in Fig. \ref{fig5}. It could be observed that the upper bound of $5000$ bits is met at any time.

\begin{figure}
\centering
\includegraphics[width=3.5in]{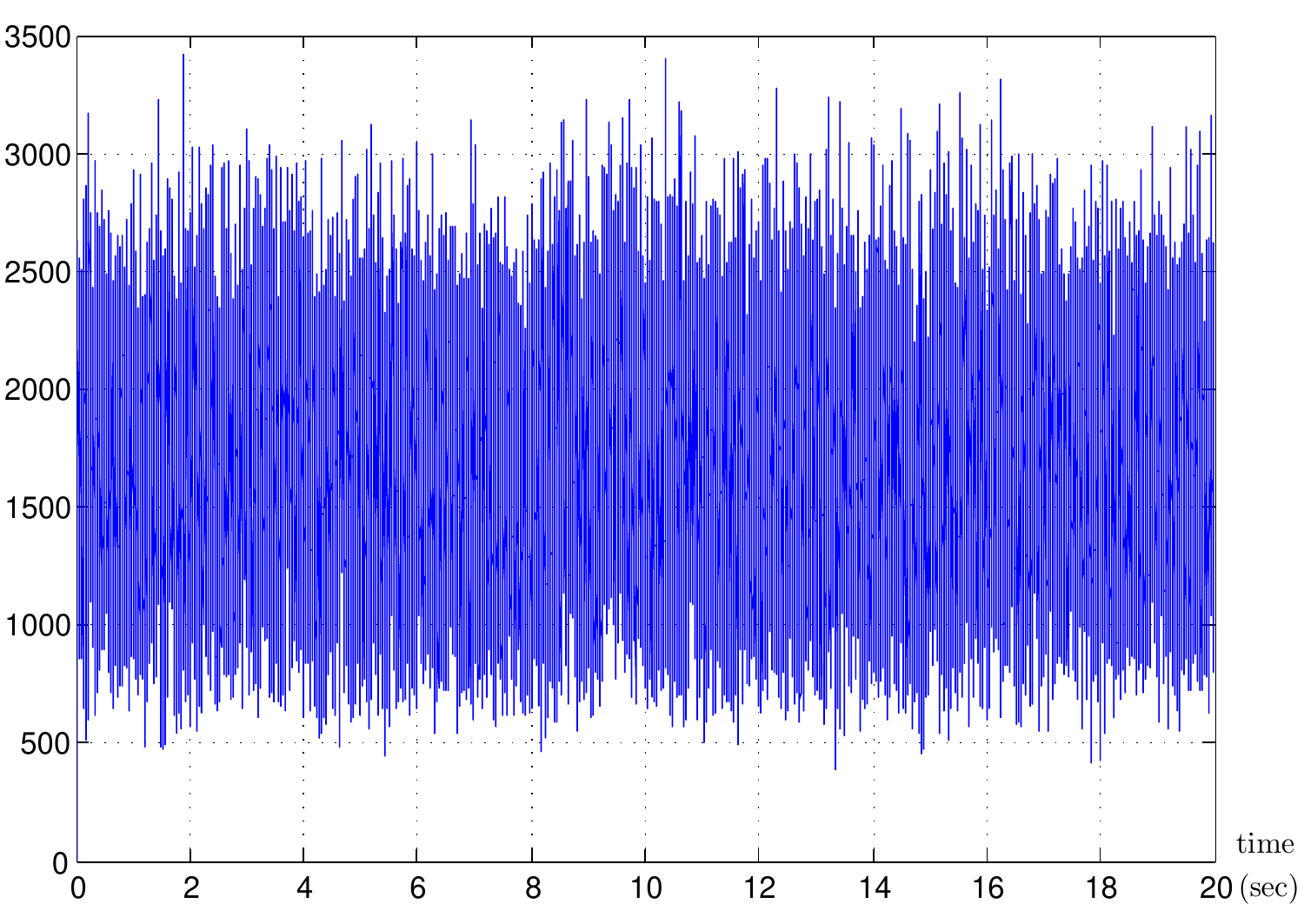}
\caption{The maximum difference (in bits) between work level of servers when all users and all servers form one cluster during interval [0,20) seconds.}
\label{fig5}
\end{figure}

Next, consider another situation where all users except user $c$ are continuously backlogged during interval $[0,20]$ seconds. Assume that user $c$ becomes backlogged at $t=10^+$ and remains continuously backlogged during interval $(10,20]$ seconds.

In the FOC corresponding to interval $[0,10]$ seconds, users $a$, $b$ and $d$ along with servers $1$ and $2$ form one cluster with service rate of $1.5$ Mbps, users $e$ and $f$ along with server $3$ form another cluster with service rate of $1.4$ Mbps, and finally users $g$, $h$, $i$ and $j$ along with server $4$ form another cluster with service rate of $1.3$ Mbps. Regarding interval $(10,20]$ seconds that all users backlogged, the FOC consists of only one cluster including all users and all servers. In this case, all users obtain a service rate of $1.25$ Mbps.

The work given to users $a$, $e$ and $g$ is shown in Fig. \ref{fig6}. It can be observed that for steady state intervals, during which the set of backlogged users does not change, the work given to any backlogged user has bounded difference with respect to its fair service share.

\begin{figure}
\centering
\includegraphics[width=3.5in]{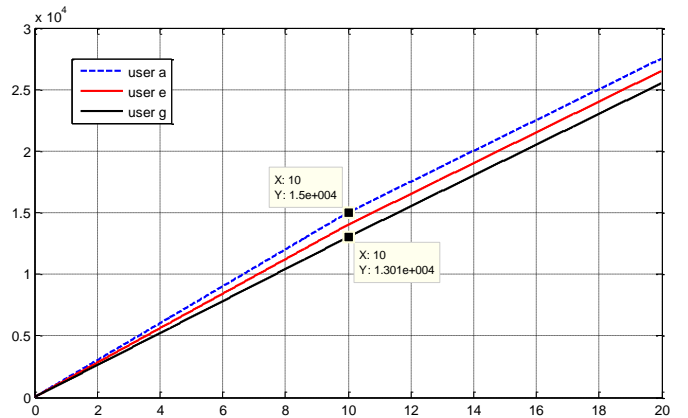}
\caption{The work given to users $a$, $e$ and $g$ during interval $[0,20)$ seconds. All users except $c$ are continuously backlogged. user $c$ becomes backlogged at $t=10^+$ second and remains backlogged afterward.}
\label{fig6}
\end{figure}

\section{Comparison with Related Work}
The system model studied in this paper has been previously studied in \cite{midrr}. Furthermore, an algorithm named as miDRR is proposed in \cite{midrr} which is built upon Deficit Round Robin (DRR) algorithm. Specifically, miDRR algorithm entails each server implementing DRR algorithm independently with a slight modification. Each server $j$ maintains one binary service flag $SF_{ij}$ for each flow $i$ that it is eligible to serve. The flag is for other servers to indicate to server $j$ that flow $i$ has been recently serviced. Maintaining this binary service flag requires two tasks \cite{midrr}.
\begin{enumerate}
  \item When server $k$ serves flow $i$, it sets service flags $SF_{ij}, \forall j\neq k$
  \item When server $j$ considers flow $i$ for service, it resets service flag $SF_{i,j}$.
\end{enumerate}

In this way, each server executes DRR algorithm independently, except that in each round only users with non-zero service flag are chosen to get service \cite{midrr}. Through a counter example, we show that this algorithm fails to achieve fairness when the length of packets are different.
\begin{example}\label{exmp6}
Consider a simple queuing system as shown in Fig. \ref{fig8a}. Assume that the weight of both users equal to 1. Let quantum size of users in DRR algorithm be equal to $Q_a=Q_b=Q$. It is assumed that flow $a$ is comprised of
fixed-sized packets of length $Q$ and flow $b$ is comprised of packets with lengths alternating between $0.75Q$, $Q$ and $0.25Q$. Assuming that server $2$ starts working immediately after server 1 at $t=0^+$, the order of scheduling packets on servers by miDRR algorithm is depicted in Figure \ref{fig8b}. It could be observed that users $a$ and $b$ get an average service rate of $4/3$ and $2/3$ Mbps respectively, which is substantially different from their $1$ Mbps fair service rate.
\end{example}

\begin{figure}
\centering
\includegraphics[width=2in]{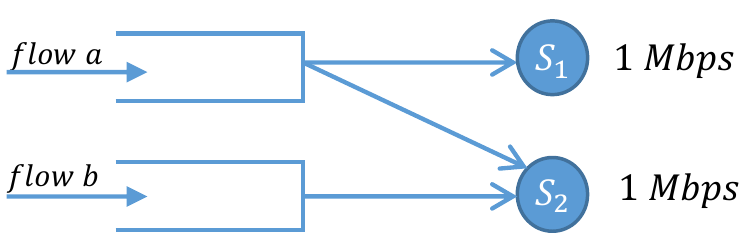}
\caption{The queuing system of example \ref{exmp6}}\label{fig8a}
\includegraphics[width=2.5in]{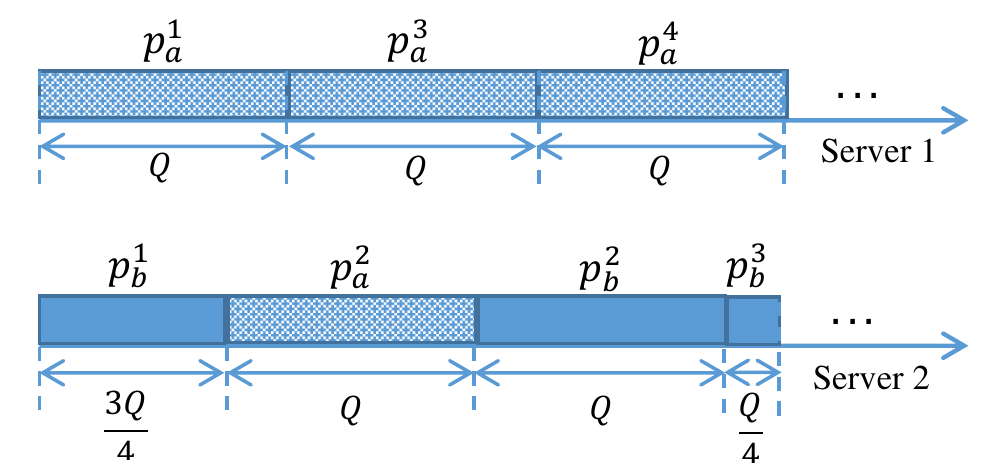}
\caption{The order of scheduling packets on servers by miDRR algorithm. This pattern will be repeated every three rounds.}\label{fig8b}
\end{figure}

\section{Conclusion}
In a constrained multi-user multi-server queuing system, max-min fair rate allocation results in multi level fair rates. Since different users have different fair service rates in such system, designing an FQ algorithm for scheduling packets on servers turns to a challenging problem. We developed $CM^4FQ$ algorithm to achieve multi level max-min fairness through a simple tagging scheme.
Specifically, $CM^FQ$ algorithm assigns a service tag to each user $i$ which represents the amount of work counted for user $i$ till time $t$. Hence, a free server will choose to serve an eligible user with the minimum service tag.
While servers' service rate is not taken into account in the algorithm, fair performance of the algorithm is shown in the sense of worst-case throughput achievable by a backlogged user. Addressing a general system model and being simply implementable, $CM^4FQ$ is applicable in a variety of practical queuing systems, especially in mobile cloud computing architecture.

\appendix
\subsection{Proof of Lemma 1}
According to $CM^4$ fairness definition, for given set of backlogged users, $B(t)=B$, there exists at least one rate allocation matrix $R(t)$ satisfying $CM^4$ fairness. Furthermore, while $R(t)$ is not necessarily unique, the normalized service rate allocated to each user will be unique \cite{DN} (otherwise it contradicts with definition of $CM^4$ fairness). We want to show that there exists a unique FOC corresponding to $B$.

The proof is by construction, i.e., we construct a FOC and we show that it is unique. To do so, consider a rate allocation matrix $R(t)$ satisfying $CM^4$ fairness. Partition the set of users into $M$ disjoint subsets, $I_m(t), 1\le m\le M$, such that for every subset $I_m(t)$ all users $i\in I_m(t)$ have the same normalized service rate (so that condition 1 in Def. \ref{def_FOC} is satisfied). Furthermore, if there exist two or more subsets with the same normalized service rate, merge them into one subset (hence, condition 2 in Def. \ref{def_FOC} is satisfied). Since the normalized service rate of each user is unique, this procedure results in a unique partitioning on the set of users.

Without loss of generality, assume that subsets of users are indexed in an increasing order of their service rate. Corresponding to each subset $I_m(t), 1\le m\le M$, create a subset of servers $S_m(t)$. Let $S_m(t)$ be set of servers giving service to users $i\in I_m(t)$. If the first subset is of zero service rate (consisting of non backlogged users), $S_1(t)$ will be the set of servers not giving service to any users.
In this way, $M$ subsets of servers are constructed which cover all servers.

Now we show that $\{S_m(t)\}_{m=1}^M$ constitutes a partitioning on the set of severs. Regarding any arbitrary $R(t)$ satisfying $CM^4$ fairness, it follows that servers giving service to users $i\in I_m(t)$ can not give service to users $j\in I_{m'}(t)$, provided that $m'>m$. Hence, servers assigned to $S_m(t)$ will not be assigned to $S_{m'}(t)$, $m'>m\ge 1$. Therefor, any tow distinct subsets $S_m(t)$ and $S_{m'}(t)$ are disjoint. Furthermore, There is no other way for assigning servers to subsets of users, while meeting the third condition in Def. \ref{def_FOC}. Hence, a unique clustering is constructed based on $\{I_m(t)\}_{m=1}^M$ and $\{S_m(t)\}_{m=1}^M$ which satisfies all conditions in Def. \ref{def_FOC}. $\blacksquare$

\subsection{Single Sever Analysis}
In this section throughput performance of $CM^4FQ$ algorithm will be analyzed when limiting number of servers to one server. This simple analysis will be useful in analyzing the general case.
\begin{lem}\label{lem7}
Consider server $S$ with constant service rate $\rho^S$. Suppose that $CM^4FQ$ algorithm is being used to give service to the set $I=\{i\}_{i=1}^N$ of $N$ users on server $S$. Defining \emph{Work Level} function of server $S$ as $V(t)=\min_i\{F_i(t)\mid i\in B(t)\}$, it can be shown that:
\begin{eqnarray}
  V(t_0,t_1) &\geq& \frac{\sum_{i\in I}W_i(t_0,t_1)}{\sum_{i\in I}\varphi_i}-\frac{(N-1).L^{max}}{\sum_{i\in I}\varphi_i} \label{eq43}\\
  &\geq& \frac{\rho^S(t_1-t_0)}{\sum_{i\in I}\varphi_i}-\frac{N.L^{max}}{\sum_{i\in I}\varphi_i}\label{eq44}
\end{eqnarray}
\end{lem}
\begin{IEEEproof}
To derive a lower bound on $V(t_0,t_1)$, consider the worst case where all users $i\in I$ are continuously backlogged during $[t_0,t_1)$. Hence, according to definition of $V(t)$, $F_i(t)\geq V(t)$ for all users $i\in I$ and for all $t\in[t_0,t_1)$. On the other hand, given that user $i$ is chosen for service at $\tau$, $F_i(\tau)= V(\tau)$. Hence, $F_i(t)$ never goes beyond $V(t)+L^{max}/\varphi_i$, i.e., $F_i(t)\le V(t)+L^{max}/\varphi_i$.

According to definition of $V(t)$, at any time $t$ there exists at least one user $i$ for which $F_{i}(t)=V(t)$. Assume that at time $t_1$, $F_{i^*}(t_1)=V(t_1)$. This along with $F_{i^*}(t_0)\geq V(t_0)$ yield:
\begin{equation}\label{eq45}
  V(t_0,t_1)\geq F_{i^*}(t_0,t_1)
\end{equation}
For other users $i\neq i^*$, $F_i(t_0)\ge V(t_0)$ and $F_i(t_1)\le V(t_1)+L^{max}/\varphi_i$. Hence:
\begin{equation}\label{eq46}
  V(t_0,t_1)\geq F_{i}(t_0,t_1)-\frac{L^{max}}{\varphi_i}
\end{equation}

multiplying both sides of (\ref{eq45}) and (\ref{eq46}) by $\varphi_{i^*}$ and $\varphi_i$ respectively, then substituting $\varphi_iF_i(t_0,t_1)$ by $W_i(t_0,t_1)$ and summing obtained inequalities over all users yield (\ref{eq43}). To show \eqref{eq44}, let $\varepsilon(t)$ denote the \emph{residual work} of the packet that is currently in service at time $t$. If the server is empty at $t$, $\varepsilon(t)$ is defined to be 0. Hence, $\sum_iW_i({t}_0,t_1)=\rho^S(t_1-{t}_0)+\varepsilon(t_1)-\varepsilon(t_0)$, which is greater than or equal to $\rho^S(t_1-{t}_0)-L^{max}$. Substituting such lower bound in (\ref{eq43}) yields (\ref{eq44}).
\end{IEEEproof}

\begin{lem}\label{lem8}
Consider a single server $S$ with constant service rate $\rho^S$. Suppose that $CM^4FQ$ algorithm is being used to give service to the set $I=\{i\}_{i=1}^N$ of $N$ users on server $S$. Let $\tilde{I}$ denote a nonempty subset of users that are continuously backlogged during $[t_0,t_1)$. It can be shown that:
\begin{eqnarray}
  V(t_0,t_1) &\leq& \frac{\sum_{i\in \tilde{I}}W_i(t_0,t_1)}{\sum_{i\in \tilde{I}}\varphi_i}+\frac{(\|\tilde{I}\|-1).L^{max}}{\sum_{i\in \tilde{I}}\varphi_i} \label{eq47}\\
  &\leq& \frac{\rho^S(t_1-t_0)}{\sum_{i\in \tilde{I}}\varphi_i}+\frac{\|\tilde{I}\|L^{max}}{\sum_{i\in \tilde{I}}\varphi_i} \label{eq48}
\end{eqnarray}
\end{lem}
\begin{IEEEproof}
The approach of the proof is almost the same as Lemma \ref{lem7}, so it is stated more concisely. To derive an upper bound on $V(t_0,t_1)$, consider the best situation that no user except users $i\in \tilde{I}$ is backlogged during $[t_0,t_1)$. Assume that for user $i^*\in\tilde{I}$, $F_{i^*}(t_0)=V(t_0)$. This along with $F_{i^*}(t_1)\geq V(t_1)$ yield:
\begin{equation}\label{eq49}
  V(t_0,t_1)\leq F_{i^*}(t_0,t_1)
\end{equation}
For any user $i\in \tilde{I}, i\neq i^*$, $F_i(t_1)\ge V(t_1)$ and $F_i(t_0)\le V(t_0)+L^{max}/\varphi_i$. Hence:
\begin{equation}\label{eq50}
  V(t_0,t_1)\leq F_{i}(t_0,t_1)+\frac{L^{max}}{\varphi_i}
\end{equation}
multiplying both sides of (\ref{eq49}) and (\ref{eq50}) by $\varphi_{i^*}$ and $\varphi_i$ respectively, then substituting $\varphi_iF_i(t_0,t_1)$ by $W_i(t_0,t_1)$ and summing obtained inequalities over all users $i\in \tilde{I}$ yield (\ref{eq47}). By defining the \emph{residual work} function, it follows that $\sum_iW_i({t}_0,t_1)=\rho^S(t_1-{t}_0)+\varepsilon(t_1)-\varepsilon(t_0)$, which is less than or equal to $\rho^S(t_1-{t}_0)+L^{max}$. Substituting such upper bound in (\ref{eq47}), yields (\ref{eq48}).
\end{IEEEproof}

\subsection{Proof of Theorem 2}
Here we prove the inequality in (\ref{eq17}). The inequality in \eqref{eq23} could be shown in the same way.
To show (\ref{eq17}), first we need to present some notation.
Assume that cluster $C=(I,S)$ consists of $K^{C}=l$ servers, $l\ge1$. Without loss of generality assume that servers $k\in S$ be indexed in an increasing order of their work level at time $t$; i.e., $V^{1}(t)\le V^{2}(t)\le ...\le V^{m}(t)\le ...\le V^{l}(t)$. Let $\hat{C}(m)\triangleq(\hat{I}(m),\hat{S}(m))$ be defined as a sub-cluster of cluster $C$, where $\hat{S}(m)\triangleq \{k\mid m\le k\le l\}$, and $\hat{I}(m)$ is defined as:

\begin{equation}\label{eq51}
  \hat{I}(m)\triangleq\{i\mid i\in I\text{ and }\pi_{i,k}=1\text{ for some }k\in \hat{S}(m)\}.
\end{equation}

Hence, the set $\bar{I}(m)\triangleq I-\hat{I}(m)$ denotes a subset of users that are eligible to get service only from the subset $\bar{S}(m)\triangleq S-\hat{S}(m)$ of servers in cluster $C$. Without loss of generality, suppose that $\bar{I}(m)\ne\emptyset$ for all $m\ge1$. According to definition of FOC, any user $i\in I$ could attain the same normalized service rate equal to $r^{C}$. Hence, for the subset $\bar{I}(m)$ of users that are eligible to get service only from the subset $\bar{S}(m)\triangleq S-\hat{S}(m)$ of servers in cluster $C$, it follows that:
\begin{equation}\label{eq52}
  \frac{\sum_{k\in\bar{S}(m)}\rho^{k}}{\sum_{i\in \bar{I}(m)}\varphi_i}\geq r^{C}.
\end{equation}
On the other hand, for users $i\in \hat{I}(m)$ it follows that:
\begin{equation}\label{eq53}
  \frac{\sum_{k\in\hat{S}(m)}\rho^{k}}{\sum_{i\in\hat{I}(m)}\varphi_i}\leq r^{C}.
\end{equation}

To show \eqref{eq17} consider the case that server $k$ has the greatest work level at time $t$, i.e., $k=l$. Let $\hat{t}$ be the last instant of time before $t$ and after $t_0$, $t_0<\hat{t}<t$, at which a packet of some user $i_0\in\hat{I}(l)$ is chosen by some server $k_0\in \bar{S}(l)$. If there is not such instant of time, let $\hat{t}=t_0$. According to this definition, $V^{l}(\hat{t})\le V^{k_0}(\hat{t})$. Given that $V^k(t_0)=V^C(t_0)$ for all servers $k\in S$, it follows that:
\begin{eqnarray}\nonumber
  V^{l}(t_0,t) & = & V^{l}(t_0,\hat{t})+V^{l}(\hat{t},t)\\
                 &\le& V^{k_0}(t_0,\hat{t})+V^{l}(\hat{t},t)\label{eq54}
\end{eqnarray}

For $V^{l}(t_0,t)$ to achieve its upper bound in (\ref{eq54}), $V^{l}(\hat{t})$ should be equal to $V^{k_0}(\hat{t})$. Suppose this to be the case. Given that $D^k({\tau}_0,\tau_1)=D^{l}({\tau}_0,\tau_1)$ for all servers $k,l\in S$ and for any subinterval $[{\tau}_0,\tau_1)$, it follows that:

\begin{eqnarray}
  w^{l}(t_0,t) = w^{k_0}(t_0,\hat{t})+w^{l}(\hat{t},t)\label{eq54_1}
\end{eqnarray}

To show (\ref{eq17}), we derive individual upper bounds on $w^{k_0}(t_0,\hat{t})$ and $w^{l}(\hat{t},t)$.
According to definition of $\hat{t}$, during interval $(\hat{t},t)$ any packets of users $i\in\hat{I}(l)$ is chosen for service only by server $l$. Hence, the results of single server analysis in Lemma \ref{lem8} imply that:
\begin{equation}\label{eq55}
V^{l}(\hat{t},t)\leq\frac{\sum_{i\in\hat{I}(l)}\varphi_iF_i(\hat{t},t)}{\Sigma_{i\in \hat{I}(l)}\varphi_i}+\frac{(\| \hat{I}(l)\|-1)L^{max}}{\sum_{i\in \hat{I}(l)}\varphi_i}.
\end{equation}

Given that $D^k({\tau}_0,\tau_1)=D^{l}({\tau}_0,\tau_1)$ for all servers $k,l\in S$ and for any subinterval $[{\tau}_0,\tau_1)$, it follows that $D_i(\hat{t},t)=D^l(\hat{t},t)$ for all users $i\in{I}$. Hence, substituting $F_i(\hat{t},t)$ from (\ref{eq12}) into (\ref{eq55}) and some manipulations results in:

\begin{eqnarray}\label{eq56}
w^{l}(\hat{t},t)\leq\frac{\sum_{i\in\hat{I}(l)}W_i(\hat{t},t)}{\Sigma_{i\in \hat{I}(l)}\varphi_i}+\frac{(\| \hat{I}(l)\|-1)L^{max}}{\sum_{i\in \hat{I}(l)}\varphi_i}.
\end{eqnarray}

Let $\varepsilon^{l}(\hat{t})$ denote the residual work of the packet in service by server $l$ at time $\hat{t}$.
According to the Theorem assumptions, let server $l$ become free at $t$ ($\varepsilon^{l}({t})=0$). This along with the fact that a packet of user $i_0\in\hat{I}(l)$ is chosen at $\hat{t}$ by server $k_0$, results in: $\sum_{i\in\hat{I}(l)}W_i(\hat{t},t)=\rho^{l}(t-\hat{t})+L^{max}-\varepsilon^{l}(\hat{t})$. To consider the worst-case, $\varepsilon^{l}(\hat{t})$ is assumed to be 0. Hence, $\sum_{i\in\hat{I}(l)}W_i(\hat{t},t)\le\rho^{l}(t-\hat{t})+L^{max}$. Substituting this upper bound in (\ref{eq56}) results in:

\begin{eqnarray}\label{eq57}
   w^{l}(\hat{t},t)\leq\frac{\rho^{l}(t-\hat{t})}{\Sigma_{i\in\hat{I}(l)}\varphi_i}+\frac{(\|\hat{I}(l)\|)L^{max}}{\sum_{i\in\hat{I}(l)}\varphi_i}
\end{eqnarray}

From (\ref{eq53}) it follows that $\rho^{l}/\Sigma_{i\in\hat{I}(l)}\varphi_i \le r^C$. On the other hand, $\sum_{i\in\hat{I}(l)}\varphi_i/\|\hat{I}(l)\|\ge\min_{i\in I}\varphi_i$. By definition of $\lambda^C$ (as in (\ref{eq18})) it follows that:
\begin{eqnarray}\label{eq58}
  \nonumber w^{l}(\hat{t},t)\leq r^{C}(t-\hat{t})+\lambda^C
\end{eqnarray}

Applying this inequality to (\ref{eq54_1}) results in:
\begin{eqnarray}
  w^{l}(t_0,t) \le w^{k_0}(t_0,\hat{t})+r^{C}(t-\hat{t})+\lambda^C\label{eq59}
\end{eqnarray}

An upper bound on $w^{k_0}(t_0,\hat{t})$ could be found easily by applying induction on the number of servers in cluster $C$, $K^C=l$. Specifically, for the case that $l=1$, there is not any server other than server $l$ from which users $i\in\hat{I}(l)$ can get service. Hence, by definition $\hat{t}=t_0$. Substituting $\hat{t}$ with $t_0$ in (\ref{eq58}) results in (\ref{eq17}) for the case of $l=1$.

Now suppose that (\ref{eq17}) holds for the case that a cluster consists of $l-1$ servers. We prove (\ref{eq17}) for cluster $C$ with $l$ servers. According to assumptions, work level of all servers are initially the same, $V^k(t_0)=V^C(t_0)$ for all servers $k\in S$. On the other hand, severs $k_0$ and $l$ have the same work level at time $\hat{t}$ and both of them becomes free at this instant ($\varepsilon^{k_0}(\hat{t})=\varepsilon^{l}(\hat{t})=0$). Therefor, servers $k_0$ and $l$ could be treated as one server when considering them during interval $[t_0,\hat{t})$. Hence, regarding this interval, cluster $C$ could be viewed as a cluster with $l-1$ servers. So, the following upper bound is concluded for $w^{k_0}(t_0,\hat{t})$, when applying (\ref{eq17}) for cluster $C$ during interval $[t_0,\hat{t})$.

\begin{eqnarray}
  w^{k_0}(t_0,\hat{t}) \le r^{C}(\hat{t}-t_0)+(l-1)\lambda^C\label{eq60}
\end{eqnarray}

Combining (\ref{eq59}) and (\ref{eq60}) results in:
\begin{eqnarray}
  w^{l}(t_0,{t}) \le r^{C}({t}-t_0)+l.\lambda^C\label{eq61}
\end{eqnarray}
$\IEEEQED$

\subsection{Proof of Theorem \ref{th2'}}
The approach of the proof is almost the same as Theorem \ref{th2}. Specifically, assume that cluster $C=(I,S)$ consists of $K^{C}=l$ servers, $l\ge1$. Without loss of generality assume that servers $k\in S$ be indexed in an increasing order of their work level at time $t$. To show \eqref{eq22} consider the case that server $k$ has the greatest work level at time $t$, i.e., $k=l$.

Following definitions presented in the beginning of the proof of Theorem \ref{th2}, let $\hat{t}$ be the last instant of time before $t$ and after $t_0$, $t_0<\hat{t}<t$, at which a packet of some user $i_0\in\hat{I}(l)$ is chosen by some server $k_0\in\bar{S}(l)$. If there is not such instant of time, let $\hat{t}=t_0$. Furthermore, let $\hat{t}_0$ be the last instant of time before $t$ and after $t_0$, $t_0<\hat{t}_0<t$, at which a packet of some user $j_0\in\hat{I}(2)$ is chosen for service by server $1$. If there is not such instant of time, let $\hat{t}_0=t_0$. Without loss of generality assume that $\hat{t}_0\le\hat{t}$. Given that server $l$ becomes free at time $t$, we can show that:
\begin{equation}\label{eq62_0}
  V^l(t_0,t)\le V^1(t_0,t)+\varsigma(\hat{t}_0,t)+(l-1)\lambda^C
\end{equation}

where $\varsigma(x)$ is a function that is $0\le\varsigma(x)\le\lambda^C$, at any time $x$, and $\varsigma(\hat{t}_0,t)=\varsigma(t)-\varsigma(\hat{t}_0)$. Given that $D^l(t_0,t)=D^1(t_0,t)$, \eqref{eq62_0} holds when:
\begin{equation}\label{eq62}
  w^l(t_0,t)\le w^1(t_0,t)+\varsigma(\hat{t}_0,t)+(l-1)\lambda^C
\end{equation}

According to \eqref{eq59} which was derived in Theorem \ref{th2}:
\begin{eqnarray}
  w^{l}(t_0,t) \le w^{k_0}(t_0,\hat{t})+r^{C}(t-\hat{t})+\lambda^C\label{eq59_1}
\end{eqnarray}

Using (\ref{eq59_1}), the upper bound in (\ref{eq62}) could be shown by applying induction on the number of servers in cluster $C$, $K^C=l$. For the case of $l=2$, there exists only two servers. Hence, $\hat{t}=\hat{t}_0$. Therefor, (\ref{eq59_1}) could be rewritten as:
\begin{eqnarray}
  w^{l}(t_0,t) \le w^{1}(t_0,\hat{t}_0)+r^{C}(t-\hat{t}_0)+\lambda^C\label{eq59_2}
\end{eqnarray}

Without loss of generality assume that server $1$ strictly attains the minimum amount of $V^k(t)$ among all servers $k\in S$ at time $t$. Such assumption implicates that sever $1$ has at least one dedicated user, i.e., there exists at least one user such as $j_1\in\bar{I}(2)$ that is only eligible to get service on server $1$. According to definition of $\hat{t}_0$, no packet of users $j\in \hat{I}(2)$ is chosen for service by server 1 during interval $(\hat{t}_0,t)$. Hence, during this interval server 1 will only choose to serve packets of users $j\in\bar{I}(2)$. Therefor, according to results of single server analysis in Lemma \ref{lem7}, it follows that:
\begin{equation}\label{}
w^{1}(\hat{t}_0,t)\geq\frac{\sum_{i\in{\bar{I}(2)}}W_i(\hat{t}_0,t)}{\Sigma_{i\in\bar{I}(2)}\varphi_i}-\frac{(\|\bar{I}(2)\|-1)L^{max}}{\sum_{i\in\bar{I}(2)}\varphi_i}
\end{equation}

Through the same line of arguments as in Lemma \ref{lem7}, it follows that: $\sum_{i\in{\bar{I}(2)}}W_i(\hat{t}_0,t)\geq\rho^{1}(t-\hat{t}_0)-L^{max}$. Substituting such lower bound in the above inequality, and also substituting the lower bound of $\rho^1$ from (\ref{eq52}) results in:
\begin{eqnarray}\label{eq63}
   w^{1}(\hat{t}_0,t)\geq{r^{C}(t-\hat{t}_0)}-\lambda^C
\end{eqnarray}

By definition of $\varsigma(x)$, (\ref{eq63}) could be rewritten as:
\begin{eqnarray}\label{eq64}
   w^{1}(\hat{t}_0,t)\geq{r^{C}(t-\hat{t}_0)}-\varsigma(\hat{t}_0,t)
\end{eqnarray}

Substituting the upper bound of $r^{C}(t-\hat{t}_0)$ from (\ref{eq64}) into (\ref{eq59_2}) results in (\ref{eq62}) for the case of $l=2$. Now, assume that (\ref{eq62}) holds for any cluster with $l-1$ servers. For cluster $C$ with $l$ servers (\ref{eq59_1}) has been derived in Theorem \ref{th2}. In deriving (\ref{eq59_1}), it has been assumed that severs $k_0$ and $l$ have the same work level at time $\hat{t}$ and both of them becomes free at this instant ($\varepsilon^{k_0}(\hat{t})=\varepsilon^{l}(\hat{t})=0$). On the other hand, work level of all servers are assumed to be initially the same, $V^k(t_0)=V^C(t_0)$ for all servers $k\in S$.

Therefor, servers $k_0$ and $l$ could be treated as one server when considering them during interval $[t_0,\hat{t})$. Hence, regarding this interval, cluster $C$ could be treated as a cluster with $l-1$ servers. Thus, the following upper bound is concluded for $w^{k_0}(t_0,\hat{t})$, when applying (\ref{eq62}) for cluster $C$ during interval $[t_0,\hat{t})$:

\begin{equation}\label{eq65}
  w^{k_0}(t_0,\hat{t})\le w^1(t_0,\hat{t})+\varsigma(\hat{t}_0,\hat{t})+(l-2)\lambda^C
\end{equation}

Applying the upper bound of $w^{k_0}(t_0,\hat{t})$ from \eqref{eq65} into (\ref{eq59_1}) results in:
\begin{eqnarray}\nonumber
  w^{l}(t_0,t) &\le& w^1(t_0,\hat{t})+(l-1)\lambda^C\\
               & + & \varsigma(\hat{t}_0,\hat{t})+r^{C}(t-\hat{t})\label{eq59_3}
\end{eqnarray}

Since $\hat{t}\ge\hat{t}_0$, the same inequality as in (\ref{eq64}) holds for interval $[\hat{t},t)$
\begin{eqnarray}\label{eq65_1}
    w^{1}(\hat{t},t)\geq{r^{C}(t-\hat{t})}-\varsigma(\hat{t},t)
\end{eqnarray}

Substituting the upper bound of $r^{C}(t-\hat{t})$ from (\ref{eq65_1}) into (\ref{eq59_3}) results in \eqref{eq62} for cluster $C$ with $K^C=l$ servers.
Given that $V^l(t_0)=V^1(t_0)$ and $\varsigma({t}_0,t)\le\lambda^C$ and from \ref{eq62_0} it follows that:
\begin{equation}\label{eq65_3}
  V^l(t)-V^1(t)\le l.\lambda^C
\end{equation}

According to results of single server analysis in Lemma \ref{lem8}, the maximum jump in work level of server $l$ at time $t$ is upper bounded by $\lambda^C$, i.e., $V^l(t^+)-V^l(t)\le\lambda^C$. Combining this with \eqref{eq65_3} results in (\ref{eq22}). $\IEEEQED$

\subsection{Proof of Theorem \ref{lem6}}
To prove the statement in Theorem \ref{lem6}, it suffices to show that for any $t\in(\hat{t}_0,t_1)$ at which one of servers $k\in S_0$ becomes free, $V^{C_m}(t)-{V}^{k}(t)>0$ for all clusters $C_m$, $m>0$. To do so, consider the worst case that for every cluster $C_m, m>0$, $r^{C_m}$ equals to $r^{C_m}\doteq r^{C_0}$ in limit ($r^{C_m}>r^{C_0}$). Hence, clusters $C_m, m>0$ could be viewed as an aggregated cluster $C'_M=(I'_M,S'_M)$.

According to the theorem assumptions $V^{C'_M}(\hat{t}_0^+)\ge \hat{V}^{C_0}(\hat{t}_0^+)+\delta$. As a result:
\begin{eqnarray}\label{eq33}
   V^{C'_M}(t)-{V}^k(t)\qquad\qquad\qquad\qquad\qquad\qquad\qquad\\
  \nonumber \geq V^{C'_M}(\hat{t}_0^+,t)-[V^k(t)-\hat{V}^{C_0}(\hat{t}_0^+)]+\delta.\!\!\!\!
\end{eqnarray}

Without loss of generality assume that $V^{C'_M}(t)<\infty$ for all $t\in(\hat{t}_0,t_1)$. Hence, all conditions of Lemma \ref{lem4} hold for cluster $C'_M$ in conjunction with interval $[\hat{t}_0^+,t)$. Hence:
\begin{eqnarray}\label{eq31}
  \nonumber V^{C'_M}(\hat{t}_0^+,t)+D^{C'_M}(\hat{t}_0^+,t)\qquad\qquad\quad\:\\
  \ge r^{C_0}(t-\hat{t}_0)-K^{C'_M}\lambda^{C'_M}.
\end{eqnarray}

Given that all users $j\in I_0$ are continuously backlogged during interval $[t_0,t_1)$, all conditions of Lemma \ref{lem5} hold for cluster $C_0$ in conjunction with interval $[\hat{t}_0^+,t)$. Given that server $k$ becomes free at time $t$, it follows that:
\begin{eqnarray}\label{eq32}
  {V}^{k}(t)-\hat{V}^{C_0}(\hat{t}_0^+)+D^{k}(\hat{t}_0^+,t)\qquad\qquad\quad\:\\\nonumber
  \leq r^{C_0}(t-\hat{t}_0^+)+K^{C_0}\lambda^{C_0}
\end{eqnarray}

The target is to show that $V^{C'_M}(t)-{V}^{k}(t)>0$. Without loss of generality assume that $V^{C'_M}(\tau)-{V}^{k}(\tau)\ge0$, for all $\tau\in(\hat{t}_0,t)$. From this premise it follows that $D^{C'_M}(\hat{t}_0^+,t)\ge {D}^{k}(\hat{t}_0^+,t)$. This inequality along with (\ref{eq26}), (\ref{eq31}) and (\ref{eq32}) result in the following inequalities:
\begin{eqnarray}\label{eq34}
\nonumber V^{C'_M}(\hat{t}_0^+,t)-[V^k(t)-\hat{V}^{C_0}(\hat{t}_0^+)]\quad\qquad\quad\\
    \geq -K^{C_0}\lambda^{C_0}-K^{C'_M}\lambda^{C'_M} > -\delta\quad
\end{eqnarray}

Combining (\ref{eq33}) and (\ref{eq34}), it follows that $V^{C_m}(t)-{V}^{k}(t)>0$ for server $k\in S_0$ which becomes free at $t$ and for all clusters $C_m, m>0$. $\IEEEQED$

\subsection{Proof of Theorem \ref{th4}}
According to Theorem \ref{lem6}, if $\min_{l>m}V^{C_l}(t_0)-\hat{V}^{C_m}(t_0)\ge\delta$, no packet of users $j\in I_l$, $l>m$, will be chosen by servers $k\in S_m$. However, given an arbitrary initial conditions at $t_0$, it will be possible for some packets of users $j\in I_l, l>m$ to be chosen by servers $k\in S_m$ during interval $[t_0,t_1)$.

Let $C'_m$ be defined as $C'_m\triangleq\{C_m,C_{m+1},...,C_M\}$.
Accordingly, let $V^{C'_m}(t)$ be defined as $V^{C'_m}(t)\triangleq\min_{l\ge m}V^{C_l}(t)$. Given that user $i$ is backlogged at $t_1$ and according to Lemma \ref{lem2_1}, $F_i(t_1)$ is greater than or equal to $V^{C'_m}(t_1)$. Defining $\Delta_i^{init}(t_0)$ as $\Delta_i^{init}(t_0)\triangleq F_i(t_0)-V^{C'_m}(t_0)$, it follows that:
\begin{equation}\label{eq66}
  F_i(t_0,t_1)\geq V^{C'_m}(t_0,t_1)-\Delta^{init}_i(t_0).
\end{equation}

Since user $i$ is continuously backlogged during $[t_0,t_1)$, it follows that $D_i(t_0,t_1)\ge D^{C'_m}(t_0,t_1)$, where $D^{C'_m}(t_0,t_1)\triangleq\min_{k\in S'_m}D^k(t_0,t_1)$. Hence, substituting $F_i(t_0,t_1)$ from (\ref{eq12}) into (\ref{eq66}) results in:
\begin{eqnarray}
  \nonumber\frac{W_i(t_0,t_1)}{\varphi_i}&\geq& V^{C'_m}(t_0,t_1)+D_{i}(t_0,t_1)-\Delta^{init}_i(t_0)\qquad\\
                                      &\geq& V^{C'_m}(t_0,t_1)+D^{C'_m}(t_0,t_1)-\Delta^{init}_i(t_0)\label{eq67}
\end{eqnarray}

As long as the minimum amount of $V^{C'_m}(t_0,t_1)+D^{C'_m}(t_0,t_1)$ is concerned, clusters $\{C_m,C_{m+1},...,C_M\}$ could be viewed as an aggregated cluster $C'_m$ with the service rate of $r^{C_m}$. Hence, all conditions of Lemma \ref{lem4} hold for cluster $C'_m$ in conjunction with the interval $[t_0,t_1)$. Therefor:
\begin{eqnarray}\nonumber
  V^{C'_m}(t_0,t_1)+D^{C'_m}(t_0,t_1)&\ge& r^{C_m}(t_1-t_0)-K^{C'_m}\lambda^{C'_m}\quad\\
                  &\ge& r^{C_m}(t_1-t_0)-K^{C'_m}\lambda_0,\label{eq68}
\end{eqnarray}
where $K^{C'_m}=\sum_{l\ge m} K^{C_l}$. Combining (\ref{eq67}) and (\ref{eq68}) results in the lower bound in (\ref{eq36}). $\IEEEQED$

\section*{Acknowledgment}
\addcontentsline{toc}{section}{Acknowledgment}
This work was started in Feb. 2015, while the first author was in Sharif University of Technology, Tehran, Iran. We wish to thank professor S. Jamaloddin Golestani for his useful comments.

\bibliographystyle{IEEEtran}
\bibliography{IEEEfull,CM4FQ}

\end{document}